\documentclass[10pt]{IEEEtran}
\usepackage{hyperref}
\usepackage{amsfonts,amssymb,amsmath,amsthm}
\usepackage[T1]{fontenc}

\newcommand\tcal{{\cal T}}

\newcommand\tr{{\rm Tr}}
\renewcommand\>{\rangle}
\newcommand\<{\langle}
\newcommand\ot{\otimes}

\newcommand\psihatab{|\hat{\psi}^+\>_{AB}}
\newcommand\psihatabprim{|\hat{\psi}^+\>_{A'B'}}
\newcommand\phitwo{{\phi_2}}
\newcommand\phione{{\phi_1}}
\newcommand\pplus{P_+}
\newcommand\portn{\tilde{P}_+^\perp}

\newcommand\suspicious{suspicious}
\newcommand\cdfs{cdfs}
\newcommand\halfp{the half-property}
\newcommand\emphHalfp{the \emph{half-property}}
\newcommand\tila{\tilde a}
\newcommand\tilb{\tilde b}

\newtheorem{lemma}{Lemma}
\newtheorem{corollary}{Corollary}
\newtheorem{theorem}{Theorem}
\newtheorem{proposition}{Proposition}
\newtheorem{definition}{Definition}
\newtheorem{observation}{Observation}
\newtheorem{fact}{Fact}

\DeclareMathOperator{\SR}{SR}
\DeclareMathOperator{\Sch}{Sch}
\DeclareMathOperator\re{Re}
\newcommand{\rhoGtwo}{{\varrho_W^\Gamma}^{\otimes2}}
\newcommand{\I}{\mathcal{I}}
\renewcommand{\P}{P_+}
\renewcommand{\O}{P_+^\bot}
\newcommand{\Xd}{{\cal X}_d}
\newcommand{\QH}{{\cal H}_Q}
\newcommand{\PH}{{\cal H}_P}
\newcommand\pQp{\<\phi_1|Q|\phi_1\>}
\newcommand\oQo{\<\phi_1^\bot|Q|\phi_1^\bot\>}
\newcommand\pQo{\<\phi_1|Q|\phi_1^\bot\>}
\newcommand\tQt{\<\phi_2|Q|\phi_2\>}

\begin{document}

\title{A Few steps more towards NPT bound entanglement}

\author{
  \L{}ukasz Pankowski$^{(1,2,4)}$
  Marco Piani$^{(2,4)*}$,
  Micha{\l} Horodecki$^{(2,4)}$,
  Pawe{\l} Horodecki$^{(3,4)}$

  \thanks{$^{(1)}$ Institute of Informatics, University of
    Gda\'nsk, Gda\'nsk, Poland}

  \thanks{$^{(2)}$ Institute of Theoretical Physics and Astrophysics,
    University of Gda\'nsk, Gda\'nsk, Poland}

  \thanks{$^{(3)}$ Faculty of Applied Physics and Mathematics,
    Gda\'nsk University of Technology, Gda\'nsk, Poland}

  \thanks{$^{(4)}$ This work is supported by EU grant SCALA
    FP6-2004-IST no.015714.}

  \thanks{$^{*}$ Present affiliation: Institute for Quantum Computing
    and Department of Physics and Astronomy, University of Waterloo,
    Waterloo ON, N2L 3G1 Canada}
}

\maketitle

\begin{abstract}
  We consider the problem of existence of bound entangled states with
  non-positive partial transpose (NPT).  As one knows, existence of
  such states would in particular imply nonadditivity of distillable
  entanglement. Moreover it would rule out a simple mathematical
  description of the set of distillable states. Distillability is
  equivalent to so called $n$-copy distillability for some $n$. We
  consider a particular state, known to be 1-copy nondistillable,
  which is supposed to be bound entangled. We study the problem of its
  two-copy distillability, which boils down to show that maximal
  overlap of some projector $Q$ with Schmidt rank two states does not
  exceed $1/2$.  Such property we call \emphHalfp.  We first show that
  the maximum overlap can be attained on vectors that are not of the
  simple product form with respect to cut between two copies.  We then
  attack the problem in twofold way: a) prove \halfp{} for {\it some
    classes} of Schmidt rank two states b) bound the required overlap
  from above for {\it all} Schmidt rank two states.  We have succeeded
  to prove \halfp{} for wide classes of states, and to bound the
  overlap from above by $c<3/4$.  Moreover, we translate the problem
  into the following matrix analysis problem: bound the sum of the
  squares of the two largest singular values of matrix $A\ot I +I \ot
  B$ with $A,B$ traceless $4\times 4$ matrices, and $\tr A^\dagger A +
  \tr B^\dagger B ={1\over 4}$.
\end{abstract}

\begin{IEEEkeywords}
  Quantum Physics, Quantum Information Theory, Bound entanglement,
  Entanglement distillation
\end{IEEEkeywords}

\section{Introduction}

The Phenomenon of bound entanglement lies at the heart of entanglement
theory \cite{rmp-ent}.  A bound entangled state of a
bipartite system is one which is entangled, but cannot be used for
quantum communication.  A possibility of transmitting qubits via
bipartite states is connected with their \emph{distillability}
\cite{BBPSSW96,BDSW1996} i.e. the possibility of obtaining asymptotically
pure maximally entangled states by local operations and classical
communication from many copies of a given state.  Such maximally
entangled states can be then used for transmitting qubits by means of
teleportation.  It is known that all entangled two qubit states are
distillable \cite{HHH1997-distill}; however, already for $3\ot 3$ or
$2\ot 4$ systems there exist bound entangled states --- entangled
states that cannot be distilled.  Such states involve irreversibility:
to create them by LOCC one needs pure entanglement
\cite{Vidal-cost2002,YangHHS2005-cost}, but no pure entanglement can
be obtained back from them.  They constitute a sort of a ``black hole''
of quantum information theory \cite{Terhal1}, and have been also
compared to a single heat bath in thermodynamics, since to create the
latter one has to spend work (as in Joule experiment), yet no work can
be obtained back from it by a cyclic process
\cite{termo,thermo-ent2002}.

Bound entangled states, although directly not useful for quantum
communication, are not entirely useless.  They can be helpful
indirectly, via activation like process: in conjunction with some
distillable state, they allow for better performance of some tasks
\cite{activation,VollbWolf_activation}.  It was even recently shown
that any bound entangled state can perform nonclassical task via kind
of activation \cite{Masanes1_activation}.  This is the first result
showing that entanglement always allows for nonclassical tasks.
Finally, it was also shown that some bound entangled states can be
useful for production of secure cryptographic key
\cite{pptkey,keyhuge,smallkey}.  This has lead to the possibility of
obtaining unconditionally secure key via channels which cannot reliably
convey quantum information \cite{GeneralUncSec,uncond-prl}.

Since bound entangled states present qualitatively different type of
entanglement from the distillable states behaving in a strange way, it
is more than desired to have some characterization of the set.  It has
been shown \cite{bound} that any state with {\it positive partial
  transpose} (PPT) \cite{Peres1} is non-distillable.  A long standing
open problem is whether the converse is also true.  Since the
discovery of bound entanglement the question ``Are all states which do
not have positive partial transpose distillable?'' has remained open.

Provided it has a positive answer, we would have computable criterion
allowing to distinguish between bound and free entanglement.  However
the importance of the problem is not merely due to technical
(in)convenience.  As a matter of fact, in \cite{ShorST2001} dramatic
consequences of a negative answer have been discovered.  Namely, for
some hypothetical bound entangled state $\varrho$ with a non-positive
partial transpose (NPT) there exists another bound entangled state
$\sigma$ such that the joint state $\varrho\ot\sigma$ is no longer a
bound entangled state.  In \cite{VollbWolf_activation} it was shown
that an arbitrary NPT bound entangled state would exhibit such a
phenomenon (it also follows from \cite{WernerPPT} via Jamio\l{}kowski
isomorphism).  Such a phenomenon of ``superactivation'' has been
indeed found in a multipartite case \cite{ShorST-superactiv} and
translated into extreme nonadditivity of multipartite quantum channel
capacities \cite{DuHoCi04}.  (In a multipartite case, though still very
strange, this can be easier to understand than in a bipartite case due
to a rich state structure allowed by many possible splits between the
parties.)  In quantum communication language the phenomenon of
``superactivation'' would mean that two channels (supported by two-way
classical communication) none of them separately can convey quantum
information if put together, can be used for reliable transmission of
qubits.  Analogous problem for channels that are not supported by
classical communication was recently solved by Smith and Yard
\cite{SmithYard} (see also \cite{CzekajPawel} in this context).
Another implication of the existence of NPT bound entangled states is that
the basic measure of entanglement --- the \emph{distillable entanglement}
--- would be non-convex.

The problem of existence of NPT bound entanglement has been attacked
many times since the beginning.  In \cite{reduction} it was shown that
it is enough to concentrate on one parameter family of the Werner states
\cite{Werner1989}:
if NPT bound entangled states exist at all, some of the Werner states must be NPT
bound entangled too.  There also exists the following characterization of
distillable states \cite{bound}: A state is distillable, if some
number of copies $\varrho^{\ot n}$ can be locally projected to obtain
a two-qubit NPT state.  The state is then called $n$-copy distillable.
Therefore, a state is non-distillable if it is not $n$-copy
distillable for all $n$.  The whole problem is to relate this rather
non operational characterization to the NPT property.

Subsequently, two attempts to solve the problem have been then made
independently \cite{DiVincenzoSSTT1999-nptbound,DurCLB1999-npt-bound}.
In particular the authors have singled out a set of the Werner states
which is expected to contain only non-distillable states.  Moreover
for any $n$ they have shown a subset of the Werner states containing solely
$n$-copy non-distillable states (see also
\cite{Bandyopadhyay2003-n-distil} in this context).  However the
subsets are decreasing
when $n$ increases.  One might ask at this point, whether $n$-copy
non-distillability implies the same for $n+1$.  Then to solve the
problem it would be enough to check whether a state is 1-copy
non-distillable, which for the Werner states is not hard to do.  However
it was shown in \cite{Watrous2003-n-dist}
that this is not true.  For any $n$ states have been found, which are
$n$-copy non-distillable, but are $(n+1)$-copy distillable.

Another way to attack the problem would be the following: let us take
a larger but mathematically more tractable class of operations than
LOCC --- the ones that preserve PPT states \cite{Rains1999,WernerPPT}.
If one can show that there are some NPT states that are not
distillable by this larger class of operations, then it would be also
true for LOCC, and the problem would be solved.  However in
\cite{WernerPPT} it have been shown that all NPT states are
distillable by PPT preserving operations.  This shows that such an
approach cannot solve our problem.

There are some sufficient conditions for distillability.  E.g. if a
state violates the reduction criterion, then it is distillable
\cite{reduction}.  In
\cite{Clarisse2005-distil,Clarisse2004-distil-maps} Clarisse provided
a systematic way of finding such conditions.  His conditions are
related to a description of the set of $1$-copy distillable states by means
of some maps and associated witnesses, in analogy to describing the set of
separable states by means of entanglement witnesses and positive maps
\cite{sep1996,TerhalReview}.  There remains the main problem of
checking such conditions on $n$-copies, to be able to prove also
$n$-copy distillability.  Another connection with separability problem
was found in \cite{KrausLC} where it was shown that the problem of
existence of NPT bound states is equivalent to showing that some
operators labeled by $n$ are entanglement witnesses.  This connection
was exploited in \cite{vianna-bound} to provide exact numerical
evidence for $2$-copy undistillability of one-copy undistillable
qutrit Werner states.

For further attempts to solve the problem see \cite{ClarissePhd} where
one can also find relevant literature.  There have been several more
recent attempts.  Unfortunately the proofs given in two of them
\cite{ChattoSarkar,SimonNPTbound} turned out to have some gaps.  The
last partial result is due to \cite{BrandaoEisert2007} where a notion
of $n$-copy correlated distillability was introduced, and used to
characterize the convex hull of the non-distillable states.

We have seen that a considerable effort has been put so far without
providing the final solution, but definitely enriching
``phenomenology'' of the problem.  In such situation we have decided
to consider a modest goal.  Namely we analyze two-copy distillability
only, and we focus on a single state, drawn from the ``suspicious''
family of the Werner states.  We choose a dimension $C^4\ot C^4$, in which
case, the problem reduces to analysis of suitable properties of some
{\it projector}.  Namely, we ask whether
\begin{align}
  \sup_{\phi_2} \<\phi_2| Q |\phi_2\>\leq {1\over 2}
\end{align}
where $Q$ is our projector on bipartite system $C^{16}\ot C^{16}$, and
supremum is taken over all states with at most two Schmidt
coefficients.  If it is true it would mean that our state would be
two-copy non-distillable.  The above condition is essentially a
special case of the condition obtained in
\cite{DiVincenzoSSTT1999-nptbound,DurCLB1999-npt-bound}.  There exists
numerical evidence that it is indeed true, however the analytical
proof is still lacking.

To begin with, we have not been able to solve even this modest
problem.  However we have obtained numerous partial results.  First of
all we have shown that the maximum overlap can be attained on vectors
that are not of the simple product form with respect to cut between
two copies.  Then we have focused research on two main approaches.
One is to provide the largest class of Schmidt rank two states
$\phi_2$ which satisfy the above inequality (a state $\phi_2$
satisfying the inequality is said to have \emphHalfp).  The other is
to provide some nontrivial bound on the quantity $\<\phi_2| Q
|\phi_2\>$.  Regarding the first approach we have provided several
classes of states satisfying \halfp{}.  In particular we have
translated the problem into a concise matrix analysis problem, and
have solved it for wide class of matrices --- normal matrices.  This
translates into a wide class of states $\phi_2$ possessing \halfp{}.
We have also shown that the problem reduces to determining whether
some family of symmetric mixed states has Schmidt number greater than
two (i.e. cannot be written as mixture of states with Schmidt rank
two).  This allows to attack the problem by means of entanglement
measures.  We have performed suitable analysis for the negativity, which
however provided smaller class of states with \halfp{} than the
previous method.

As far as the second approach is concerned, we have first analyzed the
easier problem, of supremum over {\it product} states (Schmidt rank
one).  We obtained that it gives $3/8$.  By Schwarz inequality one
obtains that the supremum over Schmidt rank two states can be at most
twice as much, giving then $3/4$.  However, as we argue, such
approach, if continued for larger number of copies, can give only
the trivial bound $1$ for $n\to \infty$.  We subsequently prove that our
quantity is for sure {\it strictly less} than $3/4$.  By continuity we
are able to push it to $\approx 0.7497$.  We also provide a couple of
other results, that may be useful for further investigation of the
problem.

The paper is organized as follows.  In section \ref{sec:problem} we
specify the main problem. In particular we introduce projector $Q$
related to two-copy distillability (and its generalizations to more
copies) and define \emphHalfp.  Then we show (Sec. \ref{sec:superpos})
that one cannot solve the problem by showing that the Schmidt rank two
states $\phitwo$ achieving the maximum are product with respect to cut
between the copies.  Subsequently (Sec.~\ref{sec:normal}) the problem
of the half-property is translated into matrix analysis problem, regarding
maximization of the sum of the squares of the two largest singular values of
matrix $A\ot I + I\ot B$ under some constraints.  We solve the problem
for normal matrices $A,B$ and obtain a wide class of states satisfying
\halfp{}. Next we show (Sec.~\ref{sec:cdf}) that any two pair state
for which at least one system from each pair is effectively two-level
one, satisfies \halfp{}. Then we turn to an easier problem of
optimizing the overlap of $Q$ with product states
(Sec.~\ref{sec:product}). We compute maximum for general case of
$n$-copies, obtaining $3/8$ for two copies. This gives bound $3/4$ for
the overlap of all Schmidt rank-two states with $Q$. We then show \halfp{}
for superpositions of the product states attaining maximum.  Then
(Sec.~\ref{sec:bounds}) we observe a trade-off between two parts of
the overlap $\<\phitwo|Q|\phitwo\>$ --- the ``diagonal'' and the
``coherence'' part, if the former is large, then the latter must be
small.  Since coherence part is bounded by diagonal one, this allows
us to go slightly below $3/4$, namely we obtain $\approx 0.7497$.
Finally we apply entanglement measures, and two-positive maps to the
problem in Sec.~\ref{sec:ent}, providing some exemplary results, which
for a while are not stronger than the ones obtained in previous
sections.  We also point that entanglement measure that would
distinguish between separable, bound entangled and distillable states
must be discontinuous.

\section{Specifying the problem}
\label{sec:problem}

It is known that if NPT bound entangled states exist then such state
must exist among the Werner states. The latter states are 
of the form
\begin{align}
  \varrho_W = p \varrho_s + (1 - p) \varrho_a
\end{align}
where
\begin{align}
  \varrho_s = \frac{P_s}{d_s}, \quad \varrho_a = \frac {P_a} {d_a}
\end{align}
with $P_s$ and $P_a$ being the projectors onto the symmetric and the antisymmetric
subspaces of the Hilbert space $\mathbb{C}^d\ot \mathbb{C}^d$ and
$d_s=d (d+1) /2$ and $d_a=d (d-1)/2$ their dimensions.  Alternatively
the Werner states may be written as
\begin{equation}
  \varrho_W = \frac{I + \alpha V}{d^2 + \alpha d}
\end{equation}
where $\alpha\in[-1,1]$ ($V=P_s - P_a$ is a swap operator).  It is
known that they are separable and PPT for $p \ge \frac12$ while for
$p < p_0 = \frac{d+1}{4d-2}$ they are distillable and for
$p\in[p_0, \frac12)$ they are NPT and it is not known whether they are
distillable. Actually it is conjectured that for the whole region
$p\in[p_0, \frac12)$ the states are NPT bound entangled
\cite{DiVincenzoSSTT1999-nptbound,DurCLB1999-npt-bound} (We will call
them the \emph{\suspicious{}} Werner states).

In \cite {bound} the characterization of the distillable states was
obtained in terms of so called $n$-copy distillability.  Namely we say
that a state is $n$-copy distillable, if $\varrho^{\ot n}$ can be
locally projected to a obtain two-qubit NPT state. Equivalently a state
$\varrho$ is $n$-copy distillable if it satisfies
\begin{equation}
  \label{eq:n-copy-distillable}
  \inf_{\phi_2} \; \<\phi_2| \varrho^{\Gamma\ot n} |\phi_2\> < 0
\end{equation}
where the infimum is taken over all pure states with Schmidt rank two,
and the superscript $\Gamma$ denotes the partial transposition.
Now a state is distillable iff it is $n$-copy distillable for some
$n$. Hence to prove that a state is non-distillable one has to show
that for all $n$
\begin{equation}
  \label{eq:n-copy-undistillable}
  \inf_{\phi_2} \; \<\phi_2| \varrho^{\Gamma\ot n} |\phi_2\> \ge 0.
\end{equation}

For the \suspicious{} Werner states it is known that they are one copy
undistillable more over it was numerically checked that they are also
two and three copy undistillable
\cite{DiVincenzoSSTT1999-nptbound,DurCLB1999-npt-bound}. As a matter
of fact for all $n$ an $n$-copy undistillable subset of the suspicious
Werner states is known, but the subsets are shrinking with $n$ giving
an empty set in the limit of $n\rightarrow \infty$.

Anyway, it is likely that even the most entangled state from the
\suspicious{} region is undistillable.  In this paper we will focus
just on this boundary state (i.e. with $p=p_0$) and moreover we
consider only the $\mathbb{C}^4 \ot \mathbb{C}^4$ case (this gives
$p=\frac5{14}$ or $\alpha=-\frac12$). The reason is that the problem
of $n$-copy distillability for the boundary state in this dimension
reduces to analyzing the overlap of rank two states with some {\it
  projector}.

Since we will be mostly concerned with two copy undistillability let
us begin with $n=2$.  The normalization of $\rhoGtwo$ has no impact on
the existence of $\phi_2$ satisfying \eqref{eq:n-copy-undistillable},
thus for $d=4$ we can simplify the expression of $\rhoGtwo$ to
\begin{align}
  \label{eq:rhoGtwo}
  \rhoGtwo &\sim {(I - {\textstyle\frac12}V)^\Gamma}^{\ot2}
  = \left(I - {\textstyle\frac{d}2} P_+\right)^{\ot2}
  \\&= (\O\ot\O + \P\ot\P) - (\O\ot\P + \P\ot\O)
\end{align}
where
\begin{equation}
  \O = I - \P, \quad \P = |\psi_+\>\<\psi_+|, \quad
  |\psi_+\> = \frac{1}{\sqrt{d}} \sum_{i=0}^{d-1}|ii\>.
\end{equation}
If we replace the minus sign with the plus sign in formula \eqref{eq:rhoGtwo} we
get the identity. Thus it is evident that two-copy undistillability,
i.e. \eqref{eq:n-copy-undistillable} with $n=2$, is equivalent to
\begin{equation}
  \label{eq:two-copy-undistillable-Q}
  \<\phi_2| Q |\phi_2\> \le \frac12
\end{equation}
for all Schmidt rank two states $\phi_2$ in the cut $AA':BB'$, or,
using a shorthand notation, for all $\phi_2 \in \SR_2(AA':BB')$, with
\begin{equation}
  Q = \O\ot\P + \P\ot\O.
\end{equation}
We will call equation \eqref{eq:two-copy-undistillable-Q} \emphHalfp.
Thus our Werner state is two copy undistillable iff all rank two
states $\phi_2$ satisfy \halfp{}.  In particular, equality in \halfp{}
\eqref{eq:two-copy-undistillable-Q} for some $\phi_2$ is equivalent to
equality in \eqref{eq:n-copy-undistillable} with $n=2$.

Thus to prove two copy undistillability we would have to show that all
two pair rank two states $\phi_2$ satisfy \halfp{}. We will show that
this is the case for a wide range of $\phi_2$ states.

We will use the notion of $\phi_k$ to denote the state of Schmidt rank
$k$ in Alice versus Bob cut. If not explicitly specified it should be
clear from the context whether we mean a state on a single pair,
i.e. $\phi_k\in\SR_k(A:B)$ or on both pairs, i.e
$\phi_k\in\SR_k(AA':BB')$.

In some cases we will consider the projector $Q$ for any dimension $d$,
though only for $d=4$ it is connected with two copy distillability of
the boundary state.

Analogously to the two copy case one can relate $n$-copy distillability of
the boundary Werner state with the overlap of rank two states with
some projectors $Q_n$. Namely for $d=4$ we have
\begin{align}
   \varrho_W^{\Gamma\ot n} &\sim (I - {\textstyle\frac{d}2}P_+)^{\ot n}
  = {\cal P}_+ - {\cal P}_-
  = I^{\ot n} - 2{\cal P}_-
\end{align}
where ${\cal P}_+$ and ${\cal P}_-$ are projectors satisfying ${\cal
  P}_+ + {\cal P}_- = I^{\ot n}$.  We define $Q_n$ as
\begin{align}
  \label{eq:Q_n}
  Q_n \equiv {\cal P}_- &= \frac12 \left(
    I^{\ot n} - \left(I - {\textstyle\frac{d}2} P_+\right)^{\ot n}
  \right)
\end{align}
so that $\<\phi_2|Q_n|\phi_2\> \le \frac12$ iff
$\<\phi_2|\varrho_W^{\ot n}|\phi_2\> \ge 0$.

\begin{lemma}
  For $d=4$ projectors $Q_n$ satisfy the following recursive formula
  \begin{align}
    Q_1 &= P_+, \\
    \label{eq:Qn-recursive}
    Q_{n+1} &= Q_n \ot Q_1^\bot + Q_n^\bot \ot Q_1.
  \end{align}
\end{lemma}

\begin{proof}
  For $n=1$ it is evident, for $n > 1$ by substituting $Q_n$
  transformed to
  \begin{align}
    \left(I - {\textstyle\frac{d}2} P_+\right)^{\ot n}
    = I^{\ot n} - 2Q_n
  \end{align}
  into $Q_{n+1}$ we obtain the recursive formula.
\end{proof}

We have $Q_2=Q$ and $Q_3$ has the form
\begin{align}
  Q_3 &= \P\ot\O\ot\O + \O\ot\P\ot\O \nonumber
  \\&\phantom{=}+ \O\ot\O\ot\P + \P\ot\P\ot\P.
\end{align}

\section{Existence of nontrivial maxima of $\tQt$}
\label{sec:superpos}

In \cite{lewenstein-2000-47-primer} a class of states of the form
$\phi_1\ot\phi_2$ was shown to provide local minimum for
\eqref{eq:n-copy-undistillable} with $d=3$, $\alpha=-\frac12$,
$n=2$. This suggests the following question: is it that all local
minima are of the form $\phi_1\ot\phi_2$?  In our specific case it
translates into the same question about the maximum.  It is easy to
see that states of the form $\phi_2\ot \phi_1$ may attain equality in
\halfp{} and nothing more.  We will now examine a question whether
there are other rank two states which attain equality in \halfp{} and
are not of this form.  The answer is unfortunately positive.

\subsection{Example of equality in superpositions}

We show that there are nontrivial superpositions of $\phi_2\ot\phi_1$
and $\phi_1'\ot\phi_2'$ which are rank two states and attain equality
in \halfp{}.

For any state of the form $\phi_2 \ot \phi_1$ its projection on $Q$
is given by
\begin{align}
  \<\phi_2\ot\phi_1| \,Q\, |\phi_2\ot\phi_1\> = p + q - 2pq \leq \frac12
\end{align}
where
\begin{align}
  p = \<\phi_2| P_+ |\phi_2\> \le \frac2d, \quad
  q = \<\phi_1| P_+ |\phi_1\> \le \frac1d
\end{align}
and the maximal value is attainable for $p=\frac2d$ and, for $d=4$, any $q$.

If we take superpositions of two states of that form with one of them
swapped
\begin{align}
  |\psi\> = \sqrt{r} |\phi_2\> \ot |\phi_1\>
  + \sqrt{1 - r} |\phi_1'\> \ot |\phi_2'\>
\end{align}
satisfying
\begin{align}
  \<\phi_2| P_+ |\phi_2\> &= \<\phi_2'| P_+ |\phi_2'\> = \frac2d, \\
  \<\phi_1| P_+ |\phi_1\> &= \<\phi_1'| P_+ |\phi_1'\> = 0
\end{align}
then
\begin{align}
  \<\psi| Q |\psi\> = \frac12.
\end{align}
States of the form $\psi$ have in general Schmidt rank higher than
two but there are also rank two states among them such as the
following class of states
\begin{equation}
  |\phi\> = \sqrt{r} \, |01\>\ot|\psi_+^2\> + \sqrt{1-r} \, |\psi_+^2\>\ot|01\>
\end{equation}
where
\begin{equation}
  |\psi_+^2\> = \frac1{\sqrt{2}} (|00\> + |11\>).
\end{equation}
The class $\phi$ can be rewritten in Alice versus Bob cut as
\begin{align}
  |\phi^{AA':BB'}\> &=
  \frac1{\sqrt{2}} |00\> \ot \left(\sqrt{r} |01\> + \sqrt{1-r} |10\>\right)
  \nonumber\\&\phantom{=}+
  \frac1{\sqrt{2}} \left(\sqrt{r} |10\> + \sqrt{1-r} |01\>\right) \ot |11\>
\end{align}
which shows that $\phi$ are rank two states in this cut.

\subsection{Form of $\phitwo$ states maximizing overlap
with $I\ot \pplus$}

In contrast to the previous section we shall show here that two pair
$\phitwo^{AA':BB'}$ which maximizes overlap with $I\ot \pplus$ must be
of the form $\phione^{A:B}\otimes \phitwo^{A':B'}$. (This result is
inspired by \cite{ChattoSarkar}) Of
course, the maximum attainable value of the projection on $P_+$ for one pair
Schmidt rank two state is $2/d$. Let $\phitwo$ be a two pair state
which attains this value.  Then we have
\begin{align}
  |\<\phitwo|\phi\>|^2= 2/d,
\end{align}
where $\phi$ is some normalized state from subspace $I\ot \pplus$,
i.e.  it is of the form
\begin{align}
  |\phi\>=\sum_j a_j |e_j f_j\>_{AB} \ot {1 \over {\sqrt d}}
  \sum_i |ii\>_{A'B'}.
  \label{eq:phi-ip}
\end{align}
Moreover also
\begin{align}
  \sup_{\phitwo\in SR2}|\<\phitwo|\phi\>|^2={2\over d}.
\end{align}
On the other hand we know that for any $\psi$
\begin{align}
  \label{eq:sum-of-coeff-sqr}
  \sup_{\phitwo\in SR2}|\<\phitwo|\psi\>|^2 = \mu_1^2+\mu_2^2,
\end{align}
where $\mu_1,\mu_2$ are the two largest Schmidt coefficients of $\psi$
in the same cut that $\phi_2$ has rank two, i.e. $AA':BB'$. Thus, as
the Schmidt coefficients of $\phi$ has the form ${a_j}/{\sqrt{d}}$ and
each of them occurs $d$ times in the composition, we have
\begin{align}
  |\<\phitwo|\phi\>|^2={2a_{\max}^2\over d},
\end{align}
where $a_{\max}=\max_j a_j$.  Therefore $a_{\max}=1$, i.e.
\begin{align}
  |\phi\> =  |x\>_A|y\>_B|\psi_+\>_{A'B'},
\end{align}
where $|x\>,|y\>$ are some states.  Writing
$|\phitwo\>=c_1|r_1\>|s_1\>+c_2|r_2\>|s_2\>$ we get
\begin{align}
  |\<\phitwo|\phi\>|^2
  ={1\over d}\left|c_1 \alpha_1 + c_2 \alpha_2\right|^2
  \le {1\over d} \left(c_1|\alpha_1| + c_2|\alpha_2|\right)^2
\end{align}
where
\begin{align}
  \alpha_1 = \sum_i({}_{AA'}\<r_1|x\>_A|i\>_{A'}) ({}_{BB'}\<s_1|y\>_B|i\>_{B'})
  \\
  \alpha_2 = \sum_i({}_{AA'}\<r_2|x\>_A|i\>_{A'})({}_{BB'}\<s_2|y\>_B|i\>_{B'})
\end{align}
Since $|\alpha_1|,|\alpha_2| \le 1$, to get $|\<\phitwo|\phi\>| =
\frac2d$ we must have $|\alpha_1|=|\alpha_2|=1$ and
$c_1=c_2=\frac1{\sqrt2}$.  It follows that $|r_1\>$ and $|r_2\>$
belong to the subspace $|x\>\<x|\ot I$. Which means that
$|r_{1(2)}\>=|x\>_A|\tilde r_{1(2)}\>_{A'}$, where $|\tilde
r_{1(2)}\>_{A'}$ are some orthogonal states.  Similar relations hold for
$|s_1\>$ and $|s_2\>$.  Thus
\begin{align}
  \phitwo = |x\>_A|y\>_B
\bigl(|\tilde r_1\>_{A'}|\tilde s_1\>_{B'}+ |\tilde
r_2\>_{A'}|\tilde s_2\>_{B'}\bigr)/\sqrt{2},
\end{align}
i.e. we obtain the desired form.

\section{States having ``normal'' projection on $Q$}
\label{sec:normal}

Here we show that if a two pair Schmidt rank two state $\phi_2$ has the
projection on $Q$ which is isomorphic to a normal operator through
a state--operator isomorphism then it satisfies \halfp{}.  To this
end we will reformulate our optimization task in terms of the two largest
Schmidt coefficients of states of the subspace defined by the
projector $Q$.  Then we will use the state--operator isomorphism to obtain
optimization problem involving matrices and finally will solve the
problem for normal matrices.

We have the following lemma, which is a generalization of a similar
one for product states \cite{AcinVC-mregs}

\begin{lemma}
  \label{lm:psi_P-form}
  For any projector $P$ acting on a bipartite system
  \begin{align}
    \sup_{\phi_2 \in \SR_2} \<\phi_2|P|\phi_2\> =
    \sup_{\psi \in \PH} (\mu_1^2 + \mu_2^2)
  \end{align}
  where $\mu_1$ and $\mu_2$ are the two largest Schmidt coefficients of
  $\psi$ and $\PH$ is the subspace defined by the projector $P$.
\end{lemma}

Note that this lemma immediately generalizes to rank $k$ states for
arbitrary fixed $k\ge1$.

\begin{proof}
  Let us observe that for all $\psi \in \PH$
  \begin{align}
    \<\phi_2|P|\phi_2\> \ge \<\phi_2|\psi\>\<\psi|\phi_2\>.
  \end{align}
  Moreover there exists $\psi \in \PH$ which reaches the equality
  \begin{align}
    \<\phi_2|P|\phi_2\> = \<\phi_2|\psi\>\<\psi|\phi_2\>,
  \end{align}
  namely $|\psi\> = \frac{P |\phi_2\>}{\|P |\phi_2\>\|}$ if
  $\|P|\phi_2\>\| \neq 0$ or any $\psi \in \PH$ otherwise.  From these
  two observations we get
  \begin{align}
    \label{eq:2P2-as-sup}
    \<\phi_2|P|\phi_2\> = \sup_{\psi \in \PH} |\<\phi_2 |\psi\>|^2.
  \end{align}
  From \eqref{eq:2P2-as-sup} and the fact stated in equation
  \eqref{eq:sum-of-coeff-sqr} we conclude
  \begin{align}
    \sup_{\phi_2 \in \SR_2} \<\phi_2|P|\phi_2\> &=
    \sup_{\psi \in \PH} \sup_{\phi_2 \in \SR_2} |\<\phi_2 |\psi\>|^2 \\
    &= \sup_{\psi \in \PH} (\mu_1^2 + \mu_2^2)
  \end{align}
  where $\mu_{1 }$ and $\mu_2$ are the two largest Schmidt coefficients of
  $\psi$.
\end{proof}

Let us now reformulate the problem in terms of matrices.  Consider the
following state--operator isomorphism
\begin{align}
  \label{eq:soi}
  |\psi\> = \sum a_{ij} |i\>|j\> \longleftrightarrow
  X = \sum a_{ij} |i\>\<j|.
\end{align}
In this isomorphism $\<\psi|\psi\> = \tr X^\dagger X$ and the Schmidt
coefficients of a state $\psi$ are equal to the singular values of the
corresponding operator $X$.  Therefore by lemma \ref{lm:psi_P-form}
and the equality between the Schmidt coefficients of $\psi$ and the singular
values of $X$ we have
\begin{align}
  \label{eq:opt-X-form}
  \sup_{\phi_2} \<\phi_2|P|\phi_2\>
  = \sup_X (\sigma_1^2 + \sigma_2^2)
\end{align}
where $\sigma_1$ and $\sigma_2$ are the two largest singular values of
operator $X$ and the supremum is taken over all operators $X$ which
correspond to states from $\PH$ through the state--operator isomorphism
\eqref{eq:soi}.

\subsection{Half-property in terms of matrices}

Let us now apply the above consideration to our particular projector
$Q$.  All states $\psi_Q \in \QH$ where $Q= \O\ot\P + \P\ot\O$ have the
form
\begin{align}
  |\psi_Q\> =
  \sqrt{p} \, |\psi_{(1)}\> |\psi_+\> +
  \sqrt{1-p} \, |\psi_+\> |\psi_{(2)}\>
\end{align}
where $p \in [0, 1]$ and
\begin{align}
  \label{eq:ort}
  |\psi_{(1)}\> \perp |\psi_+\>, \quad |\psi_{(2)}\> \perp |\psi_+\>.
\end{align}

The image of $\psi_Q$ states in the above state--operator isomorphism
have the form
\begin{align}
  X =
  \sqrt{\frac{p}{d}} \; \tilde{A} \otimes I +
  \sqrt{\frac{1-p}{d}} \; I \otimes \tilde{B}
\end{align}
where
\begin{align}
  &\tr \tilde{A} = \tr \tilde{B} = 0
  && \textrm{(orthogonality, i.e. \eqref {eq:ort})} \\
  &\tr \tilde{A}^\dagger \tilde{A} = \tr \tilde{B}^\dagger \tilde{B} = 1.
  && \textrm{(normalization)}
\end{align}
By absorbing coefficients into operators the formulation of the image
of $\psi_Q$ states can be simplified to
\begin{align}
  \label{eq:X-form}
  X = A \ot I + I \ot B
\end{align}
where
\begin{align}
  \label{eq:X-form-constraints}
  \tr A = \tr B = 0, \quad \tr A^\dagger A + \tr B^\dagger B = \frac1d.
\end{align}

Thus we have reduced the problem of \halfp{} to the following
optimization task: show that for all operators $X$ of the form
\eqref{eq:X-form} satisfying constraints \eqref{eq:X-form-constraints}
we have
\begin{align}
  \sigma_1^2 + \sigma_2^2 \le \frac12
\end{align}
where $\sigma_1$ and $\sigma_2$ are the two largest singular values of
operator $X$.

In the next section we show that this holds for normal matrices $X$
which gives a wide class of states $\phitwo$ satisfying \halfp{}.

\subsection{Half-property for states having ``normal'' projection on $Q$}

Let us first note that the operator $X$ given in equation
\eqref{eq:X-form} is normal (i.e. $X^\dagger X = XX^\dagger$)
iff operators $A$ and $B$ are normal. As
normal matrices are diagonalizable and their singular values are equal
to moduli of eigenvalues we arrive at an optimization problem over
numbers rather than matrices which we will now solve.  Namely we have

\begin{theorem}
  \label{th:Xd-halfp}
  Let $\Xd$ be a subset of normal operators $X$ of the form
  \eqref{eq:X-form} satisfying constraints
  \eqref{eq:X-form-constraints}.  Then for $d=4$ we have
  \begin{align}
    \sup_{X \in {\cal X}_d} (\sigma_1^2 + \sigma_2^2) \leq \frac12
  \end{align}
  where $\sigma_1$ and $\sigma_2$ are the two largest singular values of
  operator $X$.
\end{theorem}

\begin{proof}
  Since $X$ is diagonalizable then we can replace singular values with
  moduli of eigenvalues. The latter are of the form
  \begin{align}
    \label{eq:X-form-eigenvalues}
    \lambda_{ij} = a_i + b_j
  \end{align}
  where $a_i$ and $b_j$ are eigenvalues of $A$ and $B$ respectively.
  We then have
  \begin{align}
    \sup_{X \in \Xd} (\sigma_1^2 + \sigma_2^2)
    &= \sup_{X \in \Xd} (|\lambda_1|^2 + |\lambda_2|^2) \\
    &= \sup_{X\in\Xd}
    \max_{\substack{i,j,k,l\in\{1,\ldots,d\}, \\ (i,j)\neq (k,l)}}
    \left( |a_i + b_j|^2 + |a_k + b_l|^2 \right) \\
    &= \sup_{X\in\Xd} \max \big\{ \label{eq:two-settings}
    |a_1 + b_1|^2 + |a_2 + b_2|^2, \; \nonumber \\
    &\; \phantom{=\sup_{X\in\Xd} \max \big\{}
    |a_1 + b_1|^2 + |a_1 + b_2|^2 \big\}
  \end{align}
  where $\lambda_1$ and $\lambda_2$ are two eigenvalues of $X$ with
  largest moduli.  The constraints \eqref{eq:X-form-constraints} on $X$
  imply the following constraints on $a_i$ and $b_i$
  \begin{align}
    \label{eq:aibi-constraint-1}
    &\sum_{i=1}^d {a_i} = \tr A = 0, \quad
    \sum_{i=1}^d {b_i} = \tr B = 0, \\
    \label{eq:aibi-constraint-2}
    &\sum_{i=1}^d |a_i|^2 + \sum_{i=1}^d |b_i|^2
    = \tr A^\dagger A + \tr B^\dagger B = \frac1d.
  \end{align}
  Equality \eqref{eq:two-settings} comes from the fact that there are
  two unique settings
  \begin{enumerate}
  \item $i\neq k \wedge j \neq l$ and
  \item $i = k \wedge j \neq l \vee i \neq k \wedge j = l$.
  \end{enumerate}
  In the second setting we consider only one term of the alternative
  (as under the constraints we can exchange $A$ and $B$) and in both
  settings we take arbitrary indices (as under the constraints we can
  independently permute $a_i$ and $b_i$).

  Thus to prove the theorem we have to show that the following
  inequalities hold
  \begin{align}
    \label{eq:halfp-a1a2}
     |a_1 + b_1|^2 + |a_2 + b_2|^2 &\le \frac12 \\
     \label{eq:halfp-a1a1}
     |a_1 + b_1|^2 + |a_1 + b_2|^2 &\le \frac12
   \end{align}
   under the constraints \eqref{eq:aibi-constraint-1} and
   \eqref{eq:aibi-constraint-2} with $d=4$.  The first inequality
   comes directly from the parallelogram identity
  \begin{align}
    |x+y|^2 = 2 (|x|^2 + |y|^2) - |x-y|^2 \le 2 (|x|^2 + |y|^2)
  \end{align}
  which implies
  \begin{align}
    |a_1 + b_1|^2 + |a_2 + b_2|^2
    &\le 2 (|a_1|^2 + |b_1|^2 + |a_2|^2 + |b_2|^2 )
    \nonumber\\
    &\le 2 \frac1d = \frac12.
  \end{align}
  The second inequality is much more involved and we have moved it to
  the appendix (proposition \ref{pp:a1b1b2}) where we prove that
  \begin{align}
    |a_1 + b_1|^2 + |a_1 + b_2|^2 \le \frac{3d-4}{d^2}
  \end{align}
  which for $d=4$ gives \eqref{eq:halfp-a1a1}.
\end{proof}

We are now prepared to state the main result of this section

\begin{theorem}
  For $d=4$ any rank two state $\phi_2 \in SR_2(AA':BB')$ with
  the projection on $Q$ ($Q|\phi_2\>$) isomorphic through the
  state--operator isomorphism to a normal operator satisfies the
  half-property.
\end{theorem}

\begin{proof}
  Let us assume $\<\phi_2|Q|\phi_2\>\neq0$ (otherwise the conclusion
  is obvious).  By hypothesis $\phi_2$ reaches its projection on $Q$
  on a state $|\psi_Q\> = \frac{Q |\phi_2\>}{\|Q |\phi_2\>\|} \in \QH$
  and $\psi_Q$ is isomorphic through the state--operator isomorphism
  given by \eqref{eq:soi} to a normal operator $X$.  Then using the
  fact stated in equation \eqref{eq:sum-of-coeff-sqr}, equality of
  the Schmidt coefficients of $\psi_Q$ and the singular values of operator $X$
  in the state--operator isomorphism, and theorem \ref{th:Xd-halfp} we
  obtain
  \begin{align}
    \<\phi_2|Q|\phi_2\> &= |\<\phi_2|\psi_Q\>|^2
     \le \sup_{\phi_2\in \SR_2(AA':BB')} |\<\phi_2|\psi_Q\>|^2 \\
     &= \mu_1^2 + \mu_2^2 = \sigma_1^2 + \sigma_2^2
     \le \sup_{X \in {\cal X}_d} (\sigma_1^2 + \sigma_2^2)
     \le \frac12
  \end{align}
  where $\mu_1$ and $\mu_2$ are the two largest Schmidt coefficients of
  $\psi_Q$ in the same cut in which $\phi_2$ has rank two (i.e
  $AA':BB'$) while $\sigma_1$ and $\sigma_2$ are the two largest singular
  values of operator $X$, and $\Xd$ is a subset of normal operators
  $X$ of the form \eqref{eq:X-form} satisfying constraints
  \eqref{eq:X-form-constraints}.
\end{proof}

\subsection{Characterization of states with normal projection onto $Q$}
\label{sec:matrix-C}

A more operational characterization of the states for which the above
theorem proves \halfp{} is the following.  Suppose we project $\phi_2$
state onto $\psi_+$ on subsystem $AB$.  Then the subsystem $A'B'$
should collapse to a $*$-symmetric state, i.e. a state of the form
\begin{align}
  \sum a_i |e_i\>_{A'} |e_i^*\>_{B'}.
\end{align}
The same should hold for the projection on $A'B'$.

To see it let us use the state--operator isomorphism \eqref{eq:soi}. In
our particular case it will read as follows
\begin{align}
  |\phitwo\>=(C_{AA'}\ot I_{BB'}) \psihatab\ot \psihatabprim
\end{align}
with $\hat\psi_+=\sum_i |ii\>$, or simply
\begin{align}
  |\phi_2\>=\sum_{i,i',j,j'}
  C_{ii'\, jj'} |ii'\>_{AA'}|jj'\>_{BB'}.
\end{align}
We will further write $\phi_2 \propto C$.  If for an example the matrix
$C$ is normal the corresponding state is of the form
\begin{align}
  |\phitwo\> = a|e\>_{AA'}|e^*\>_{BB'}+ b|f\>_{AA'}|f^*\>_{BB'}
\end{align}
where $e\perp f$.  Here $a$ and $b$ are eigenvalues of $C$, hence
Hermitian $C$ means that they are real, while positive $C$ matrix
means that $a$ and $b$ are nonnegative.  (We have only two terms
because $\phitwo$ is of Schmidt rank two).

Let us now examine the projection of $\phitwo$ onto $\QH$. We have
\begin{align}
  Q|\phi_2\>&=|\psi^+\>_{AB}\ot
  \left(|\tilde\phi^{(2)}\>_{A'B'}-{1\over d}\tr C\, |\psi^+\>_{A'B'}\right)
  \nonumber\\&\phantom{=}+
  \left(|\tilde\phi^{(1)}\>_{AB}-{1\over d}\tr C\, |\psi^+\>_{AB}\right)
  \ot |\psi^+\>_{A'B'}
\end{align}
where
\begin{align}
  &|\tilde\phi^{(2)}\>_{A'B'}=
  {}_{AB}\<\psi_+|\phi_2\>
  \propto {1\over d}C_{A'}\\
  \label{eq:tilde-phi-1}
  &|\tilde\phi^{(1)}\>_{AB}=
  {}_{A'B'}\<\psi_+|\phi_2\>
  \propto {1\over d} C_A
\end{align}
are unnormalized states that are obtained on one pair after projecting
second pair onto maximally entangled state $P_+$; here
$C_A=\tr_{A'}C_{AA'}$, $C_{A'}=\tr_{A}C_{AA'}$.  Let us now relate
$C_A$ and $C_{A'}$ with the matrices $A$ and $B$ from
\eqref{eq:X-form}.  Thus partial traces of matrix $C_{AA'}$ correspond
to unnormalized states that emerge after projecting one pair onto
$P_+$.

The projection of $\phi_2$ onto $\QH$ can be also written as follows
\begin{align}
  Q|\phitwo\>=|\psi^+\>_{AB}\ot |\phi^{(2)}\>_{A'B'} +
  |\phi^{(1)}\>_{AB} \ot |\psi^+\>_{A'B'}
\end{align}
where
\begin{align}
  &|\phi^{(1)}\>_{AB}=(Y_A\ot I)\psihatab\\
  \label{eq:phi^(2)}
  &|\phi^{(2)}\>_{A'B'}=(Y'_{A'}\ot I)\psihatabprim
\end{align}
with
\begin{align}
  Y={1\over d} C_A - {\tr C \over d^2}I_A; \quad
  Y'={1\over d} C_{A'} - {\tr C \over d^2}I_{A'}.
\end{align}
(Note that $Y$ and $Y'$ are traceless, which means that corresponding
vectors are orthogonal to $\psi_+$). We see that---up to a factor---$A$ is
equal to $Y$ and $B$ is equal to $Y'$.  Now since we assume that $A$
and $B$ are normal then $C_A$ and $C_{A'}$ must also be normal. This
means that e.g. $C_A$ is of the form
\begin{align}
  C_A = \sum_i c_i |e_i\>\<e_i|
\end{align}
where $c_i$ are complex numbers and $e_i$ form an orthonormal basis.
Thus the state \eqref{eq:tilde-phi-1} coming from projecting subsystem
$A'B'$ onto $P_+$ will have the desired form
\begin{align}
  \sum_i a_i |e_i\>_A|e_i^*\>_{B},
\end{align}
and similarly for projecting $AB$ part onto $P_+$.

\section{Half-property for low Schmidt rank states}
\label{sec:cdf}

In this section we show that any state which on each pair has at least
one subsystem with one-qubit support satisfies \halfp{}. To this end
we will use the notion of so called \emph{common degrees of freedom}
introduced in the following subsection.

\subsection{Half-property via ``common degrees of freedom''}

We begin with the following definition

\begin{definition}
  For a given state $\phi$ we define a set called \emph{common degrees
    of freedom} of subsystems $A$ and $B$ as
  \begin{equation}
    \operatorname{cdf}(\phi, A, B) = \{i \in \I: \<\phi|P_i|\phi\> \neq0 \}
  \end{equation}
  where $\I = \{0,\ldots,d-1\}$ and
  \begin{align}
    P_i &= |ii\>\<ii|_{AB} \ot I_{A'B'}.
  \end{align}
  We say that subsystem $A$ has at most $k$ common degrees of freedom
  with subsystem $B$ if $|\operatorname{cdf}(\phi, A, B)| \le k$.
\end{definition}

\begin{proposition}
  \label{pp:cdf-and-halfp}
  If for a given state $\phi$ subsystems $A$ with $B$ and $A'$ with
  $B'$ have at most $\frac{d}{2}$ common degrees of freedom then
  $\phi$ satisfies \halfp{}.
\end{proposition}

\begin{proof}
  We will show that if for a given state $\phi$ subsystems $A$ with
  $B$ and $A'$ with $B'$ have at most $\frac{d}{2}$ common degrees of
  freedom then
  \begin{equation}
    \<\phi|Q|\phi\> = \frac12 \<\phi|\tilde{Q}|\phi\> \le \frac12
  \end{equation}
  where $\tilde{Q}$ is some other projector.

  Let us define
  \begin{align}
    P_d &= \frac{1}{d} \sum_{i,j \in \I} |ii\>\<jj|, \\
    P_{AB} &= \frac{2}{d} \sum_{i,j \in \I_{AB}} |ii\>\<jj|
    &&\text{with~} |\I_{AB}| = \frac{d}{2}
    \\&&&\text{and~}\nonumber
    \operatorname{cdf}(\phi,A,B)\subset \I_{AB} \subset \I
    \\
    P_{A'B'} &= \frac{2}{d} \sum_{i,j \in \I_{A'B'}} |ii\>\<jj|
    &&\text{with~} |\I_{A'B'}| = \frac{d}{2}
    \\&&&\text{and~}\nonumber
    \operatorname{cdf}(\phi,A',B')\subset \I_{AB} \subset \I
  \end{align}
  where $P_d$ is a maximally entangled state in $d \ot d$. $P_{AB}$
  and $P_{A'B'}$ are maximally entangled states on $\frac{d}{2}
  \otimes \frac{d}{2}$ subspaces chosen in such a way to contain
  common degrees of freedom of $A$ with $B$ and $A'$ with $B'$
  respectively. $\I_{AB}$ and $\I_{A'B'}$ are extensions of the sets
  of common degrees of freedom (with whatever elements) to get sets of
  exactly $\frac{d}{2}$ elements.

  One can observe that in the expression
  \begin{equation}
    \<\phi|P_d^{AB} \ot I^{A'B'} |\phi\>
  \end{equation}
  $\phi$ projects only onto those $|ii\>\<jj|$ of $P_d$ for which $i,j
  \in \operatorname{cdf}(\phi,A,B)$ by the very definition of common
  degrees of freedom, thus we can remove any of $|ii\>\<jj|$ having $i
  \notin \operatorname{cdf}(\phi,A,B)$ or $j \notin
  \operatorname{cdf}(\phi,A,B)$ in particular we can remove all those
  for which $i \notin \I_{AB}$ or $j \notin \I_{AB}$ which gives us
  \begin{equation}
    \<\phi|P_d^{AB} \ot I^{A'B'} |\phi\> =
    \<\phi|\frac12 P_{AB} \ot I^{A'B'} |\phi\>
  \end{equation}
  similar consideration for other elements of $Q$ gives us
  \begin{align}
    \<\phi|Q|\phi\>
    &= \<\phi| I \ot P_d^{A'B'} + P_d^{AB} \ot I - 2 P_d^{AB} \ot P_d^{A'B'} |\phi\>
    \\
    &= \<\phi| I \ot \frac12 P_{A'B'} + \frac12 P_{AB} \ot I
    - 2 \frac12 P_{AB} \ot \frac12 P_{A'B'} |\phi\> \\
    &= \frac12 \<\phi| I \ot P_{A'B'} + P_{AB} \ot I - P_{AB} \ot P_{A'B'} |\phi\>
    \\
    &= \frac12 \<\phi| \tilde{Q} |\phi\> \le \frac12
  \end{align}
  where $\tilde{Q}$ is also a projector thus the inequality holds.
\end{proof}

\subsection{Example: states with positive matrix $C$}

We begin by rephrasing number of \cdfs\  in terms
of the matrix $C$ of a state (see sec. \ref{sec:matrix-C})
written in block form:
\begin{align}
  C_{AA'}=\sum_{ij} |i\>_{A}\<j| \ot C^{ij}_{A'}.
\end{align}
The number of \cdfs\ is the number of blocks $C^{(ii)}$, i.e. {\it
  diagonal} blocks which do not vanish (i.e. which have at least one
nonzero element). The proposition \ref{pp:cdf-and-halfp} says that for
any given state (not necessarily of Schmidt rank two) the number of
\cdfs\ is less than or equal to $2$, then the state has \halfp.

Now suppose that $C$ is positive. Then the diagonal blocks are
positive matrices, and they do not vanish iff their trace is
nonzero. Thus the full information about the number of \cdfs\ is
contained in the partial trace of the matrix $C$:
\begin{align}
  C_{A}=\tr_{A'} C_{AA'}=\sum_{ij}\tr(C^{ij}_{A'}) |i\>_{A}\<j|
\end{align}
Thus number of \cdfs\ is equal to the number of nonzero elements on the
diagonal of $C_A$.

Now, since $Q$ is invariant over pairwise $U\ot U^*$ transformations,
we can rotate a state to diminish the number of \cdfs{} as much as possible.
If we can get $2$ or less, then we obtain \halfp.  Consider e.g. such
transformation for the pair $AB$.  The matrix $C_A$ then transforms as
$U C_AU^\dagger$.  We are interested in the minimal number of nonzero
diagonal elements under such transformations, which equals to the rank
of the matrix $C_A$.  We have then obtained, that any state with positive
matrix $C$ such that its partial trace has rank $\leq 2$, has \halfp.

Let us note however that our result of section \ref{sec:matrix-C}
implies that all Schmidt rank two states with positive matrix $C$
satisfy \halfp{}.

\subsection{Application of cdf to low Schmidt rank}

Here by use of proposition \ref{pp:cdf-and-halfp} we show that any
state which on each pair has at least one subsystem with one-qubit
support satisfies \halfp{}.

\begin{theorem}
  \label{th:sch-halfp}
  Any state $\phi$ that satisfies
  \begin{align}
    \left(
      \Sch(A:A'BB') \le \frac{d}{2} \vee
      \Sch(B:AA'B') \le \frac{d}{2}
    \right)
    \nonumber\\ \; \wedge \;
    \left(
      \Sch(A':ABB') \le \frac{d}{2} \vee
      \Sch(B':AA'B) \le \frac{d}{2}
    \right)
  \end{align}
  also satisfies \halfp{}. Here $\Sch(X:Y)$ denotes the Schmidt
  rank of the state $\phi$ in the $X$ versus $Y$ cut.
\end{theorem}

\begin{observation}
  \label{obs:Q-invariance}
  The operator $Q$ is $U_A \ot V_{A'} \ot U_{B}^* \ot V_{B'}^*$
  invariant. (Where $U$ and $V$ are unitaries).
\end{observation}

\begin{proof}[Proof of theorem \ref{th:sch-halfp}]
  The hypothesis may be expanded into a four-term alternative.  We
  prove the conclusion for one of the terms (for the others the proof is
  analogous).  Now suppose
  \begin{equation}
    \Sch(A:A'BB') \le \frac{d}{2} \; \wedge \;
    \Sch(A':ABB') \le \frac{d}{2}
  \end{equation}
  which means that there are Schmidt decompositions of $\phi$ of the
  form
  \begin{equation}
    |\phi\>
    = \sum_{i=0}^{d/2 - 1} a_i |\psi_i^A\> |\psi_i^{A'BB'}\>
    = \sum_{i=0}^{d/2 - 1} a'_i |\psi_i^{A'}\> |\psi_i^{ABB'}\>
  \end{equation}

  We can choose such $U$ and $V$ which transform $\phi$ to
  \begin{align}
    |\phi'\> &= U_A \ot V_{A'} \ot U_{B}^* \ot V_{B'}^*  |\phi\>
    \\&
    = \sum_{i=0}^{d/2 - 1} a_i |i^A\> |\tilde\psi_i^{A'BB'}\>
    = \sum_{i=0}^{d/2 - 1} a'_i |i^{A'}\> |\tilde\psi_i^{ABB'}\>
  \end{align}
  Now we can observe that $A$ with $B$ and $A'$ with $B'$ have at most
  $\frac{d}{2}$ degrees of freedom in common in $\phi'$ (as there are
  clearly at most $\frac{d}{2}$ degrees of freedom on $A$ and $A'$
  subsystems) thus by applying proposition \ref{pp:cdf-and-halfp} we
  have
  \begin{equation}
    \<\phi'| Q |\phi'\> \le \frac12
  \end{equation}
  and by applying observation \ref{obs:Q-invariance} we finally get
  \begin{equation}
    \<\phi| Q |\phi\> = \<\phi'| Q |\phi'\> \le \frac12.
  \end{equation}
\end{proof}

\section{Optimizing over product states and implications}
\label{sec:product}

In this section we will first consider a simpler question from the
original one.  Namely we will optimize the overlap of $Q$ with product
states rather than with Schmidt rank two ones.  This is equivalent to
optimization of the overlap of $Q^\Gamma$ with product states, where
$Q^\Gamma$ is the partial transpose of $Q$.  We find the maximal overlap with
product states for the general case of $n$ copies i.e. we will work
with $Q_n$ given by \eqref{eq:Q_n}.  Knowing the maximum over product
states, we can bound the maximum over Schmidt rank two states. For
$n=2$ we will obtain in this way
\begin{align}
  \<\phi_2| Q |\phi_2\> \le \frac34.
\end{align}
However the analysis of $n$ copy case shows that in the limit of $n\rightarrow\infty$ one
obtains a trivial result that the overlap does not exceed one.
Nevertheless this approach will be used in subsequent section to go
beyond $\frac34$.  Analysis of $Q^\Gamma$ also allows for direct proof
of \halfp{} for states with positive matrix $C$.

\subsection{Maximum overlap of product states with $Q_n$}

To find the maximum overlap of product states with $Q_n$ given by \eqref
{eq:Q_n} we will first analyze spectral decomposition of $Q_n^\Gamma$.
We have
\begin{align}
  Q_n^\Gamma
  &= \frac12\textstyle\left(I^{\ot n} - \left(
      I - \frac12 V\right)^{\ot n} \right) \\
  &= \frac12\textstyle\left(I^{\ot n} - \left(
      \frac12 P_s + \frac32 P_a\right)^{\ot n} \right) \\
  &= \sum_{i=0}^n \lambda_i A_i
\end{align}
where $P_s$ and $P_a$ are the projectors onto the symmetric and the antisymmetric
subspaces and
\begin{align}
 \lambda_i &=\frac12 \left(1 - \frac{3^i}{2^n}\right) \\
 A_i &= \sum_{l_j\in\{0,1\}, \; \sum l_j = i}
 a_{l_1}\ot \cdots \ot a_{l_n}
\end{align}
with $a_0 = P_s$ and $a_1=P_a$.  (Note that $\sum_{i=0}^n A_i =
I^{\ot{n}}$).  Thus eigenvalues of $Q_n^\Gamma$ are in decreasing
order and the largest eigenvalue $\lambda_0$ is associated with the
eigenspace $A_0=P_s^{\ot n}$.  In particular for $n=2$ we have
\begin{align}
  \lambda_0= \frac38,\, \lambda_1 = \frac18,\, \lambda_2 = - \frac58,
\end{align}
so that
\begin{align}
  Q_2^\Gamma={3\over 8} P_s\ot P_s - {5\over 8} P_a\ot P_a
  +{1\over 8}(P_a\ot P_s +P_s\ot P_a).
\end{align}
Let us now compute the maximum overlap of product states with $Q_n$.
Since $(\tr Q_n|\phi_1\>\<\phi_1|)^\Gamma=\tr Q_n^\Gamma
|\tilde\phi_1\>\<\tilde\phi_1| $, where $\tilde\phi_1$ is also a
product state (with a one-to-one correspondence between $\phi_1$ and
$\tilde\phi_1$), we can replace the optimization on $Q_n$ with
an optimization on $Q_n^\Gamma$. The overlap of product states with
$Q_n^\Gamma$ is bounded by its largest eigenvalue $\lambda_0$ and this
bound is attainable as in the eigenspace $P_s^{\ot{n}}$ corresponding to
$\lambda_0$ there are product states.  We thus have
\begin{align}
  \label{eq:r1-max}
  \sup_{\phi_1} \<\phi_1|Q_n|\phi_1\>
  = \sup_{\phi_1} \<\phi_1|Q_n^\Gamma|\phi_1\>
  = \lambda_0 = \frac12\left(1 - \frac1 {2^n}\right).
\end{align}
In particular for two copies this gives $\frac38$.

\subsection{Bound for $\tQt$ in terms of $\pQp$}

As Schmidt rank two state may be decomposed to
\begin{align}
  |\phi_2\> = \sqrt {p} |\phi_1\> + \sqrt {1-p}|\phi_1^\bot\>,
\end{align}
we observe that
\begin{align}
  \nonumber
  &\sup_{\phi_2} \<\phi_2|Q|\phi_2\> \\
  &= \sup_{\phi_1,\phi_1^\bot,p}
  (\sqrt {p} \<\phi_1| + \sqrt {1-p}\<\phi_1^\bot|) Q
  (\sqrt {p} |\phi_1\> + \sqrt {1-p}|\phi_1^\bot\>) \\
  \nonumber
  &= \sup_{\phi_1,\phi_1^\bot,p}
  p \pQp + (1-p) \oQo \\&\phantom{+} + 2\sqrt {p (1 - p)} \re\pQo \\
  \label {eq:sup:pQp+pQo}
  &\le \sup_{\phi_1,\phi_1^\bot} (\pQp + |\pQo|)
\end{align}
and thus from Schwarz inequality
\begin{align}
  \sup_{\phi_2} \tQt \le 2\sup_{\phi_1} \pQp.
\end{align}
In this way we have obtained the bound for the overlap of the Schmidt rank two
states with $Q$ in terms of optimal overlap with product states.  This
is also true for any other projector, in particular, for $Q_n$.

Thus for two copies we obtain the following bound
\begin{align}
  \sup_{\phi_2} \tQt \le \frac34.
\end{align}
Unfortunately this method does not lead to any bound that would hold
for all $n$ apart from the trivial bound $\<\phi_2|Q_n|\phi_2\>\le 1$.

\subsection{The form of the rank-one states attaining maximum on $Q_n$}

It is interesting that the product states attaining the maximum on $Q_n$
must be of a very specific form.  For $n=2$ the partial transpose of such
state (which is again a legitimate state) must belong to a subspace
$P_s^{AB} \ot P_s^{A'B'}$.  One can then find that the states that are
product with respect to $AA':BB'$ cut and the same time belong to the
above subspace must be of the form
\begin{align}
  |xx\>_{AB} \ot |yy\>_{A'B'}.
\end{align}
It then follows that a product state maximizing overlap with $Q_n$
must be of the form
\begin{align}
  |xx^*\>_{AB} \ot |yy^*\>_{A'B'}.
\end{align}

This observation in general case of $n$ copies is contained in the
following.

\begin{proposition}
  \label{pp:r1-form}
  For any $n$ all rank-one states $\phi_1$ reaching maximum on $Q_n$
  has the form
  \begin{align}
    \label{eq:form}
    |\phi_1\> = \bigotimes_{i=1}^n |\psi_i\>_{A_i}|\psi_i^*\>_{B_i}.
  \end{align}
\end{proposition}

\begin{proof}
  The thesis of the proposition is equivalent to the following
  statement: for any $n$ all rank-one states $\phi_1$ reaching maximum
  on $Q_n^\Gamma$ have the form
  \begin{align}
    \label{eq:gamma-form}
    |\phi_1\> = \bigotimes_{i=1}^n |\psi_i\>_{A_i}|\psi_i\>_{B_i}.
  \end{align}
  We prove it by induction.
  \begin{enumerate}
  \item For $n=1$ only rank one states of the form $|\psi\psi\>$ reach
    maximum on $Q_1^\Gamma=\frac{1}{4}V$.
  \item Suppose for some $n$ maximal projection of rank one state on
    $Q_n^\Gamma$ requires the from \eqref{eq:gamma-form}. From
    previous section a rank one state $\phi_1$ defined on $n+1$ pairs
    to attain maximum on $Q_{n+1}^\Gamma$ must be an eigenstate of
    $P_s^{\otimes n+1}$ which is a subspace of the symmetric space on
    $n+1$ pairs.  Thus the Schmidt decomposition of $\phi_1$ in $n$ pairs
    versus single pair cut ($AB:ab$) has the form
    \begin{align}
      |\phi_1\> = |\psi\>_{Aa}|\psi\>_{Bb}
      = \sum a_i a_j |\psi_i\psi_j\>_{AB}|\phi_i\phi_j\>_{ab}
    \end{align}
    and we have
    \begin{align}
      \nonumber
      &\<\phi_1|P_s^{\otimes n+1}|\phi_1\> \\
      &= \sum a_i a_j a_k a_l \<\psi_i\psi_j|P_s^{\otimes n}|\psi_k\psi_l\>
      \<\phi_i\phi_j|P_s|\phi_k\phi_l\> \\
      &= \sum a_i a_j a_k a_l \<\psi_i\psi_j|P_s^{\otimes n}|\psi_k\psi_l\>
      \frac12 (\delta_{ik}\delta_{jl} + \delta_{il}\delta_{jk})
    \end{align}
  \end{enumerate}
  to obtain one above all the projections must be equal to 1.  For
  projection on $P_s$ given in delta-form requires $i=j=k=l$ and it is
  always one only if $\phi_1$ is product in $AB:ab$ cut. To obtain one
  on $P_s^{\otimes n}$ the $\psi_i\otimes\psi_i$ state must be of the
  form \eqref{eq:gamma-form} and thus $\phi_1$ is of the form
  \eqref{eq:gamma-form}.
\end{proof}

\subsection{Superpositions of rank-one states with maximum on $Q_n$}

One could expect that superpositions of rank-one states with maximum
on $Q_n$ has the \halfp{} as such rank-one states are product
between the copies.  Indeed this is the case, their overlap with $Q_n$
is analyzed in the following

\begin{proposition}
  Let $d=4$ and $\phi_1$, $\phi_1^\bot$ be $n$-copy orthogonal product
  states with maximum overlap with $Q_n$, i.e. of the form
  \begin{align}
    |\phi_1\>
    = \bigotimes_{i=1}^n |\psi_i\>_{A_i}|\psi_i^*\>_{B_i}, \quad
    |\phi_1^\bot\>
    = \bigotimes_{i=1}^n |\tilde\psi_i\>_{A_i}|\tilde\psi_i^*\>_{B_i}
  \end{align}
  then their superposition
  \begin{align}
    |\phi_2\> = \sqrt{p} |\phi_1\> + \sqrt{1-p} |\phi_1^\bot\>
  \end{align}
  has the following overlap with $Q_n$
  \begin{align}
    \label{eq:r1-max-superpos}
    \<\phi_2|Q_n|\phi_2\>
    = \frac12 \left(1 - \frac1{2^n}\right)
    - \sqrt{p (1 - p)} \prod_{i=1}^n \left(
      |\<\psi_i|\tilde\psi_i\>|^2 - \frac12\right).
  \end{align}
  In particular it is equal to $\frac12$ only if $p=\frac12$ and
  $\phi_1$, $\phi_1^\bot$ are orthogonal on an odd number of copies
  and equal on the rest. Otherwise it is less than $\frac12$.
\end{proposition}

\begin{proof}
  The form of $\phi_1$ and $\phi_1^\bot$ comes from proposition
  \ref{pp:r1-form} and their overlap with $Q_n$ from \eqref{eq:r1-max}
  thus we have
  \begin{align}
    \<\phi_2|Q_n|\phi_2\>
    = \frac12 \left(1 - \frac1{2^n}\right)
    + 2 \sqrt{p (1 - p)} \re\<\phi_1|Q_n|\phi_1^\bot\>
  \end{align}
  Thus to finish the proof we will show by induction that
  \begin{align}
    \label{eq:r1-max-coh}
    \<\phi_1|Q_n|\phi_1^\bot\>
    = -\frac12 \prod_{i=1}^n \left(
      |\<\psi_i|\tilde\psi_i\>|^2 - \frac12\right)
  \end{align}
  It is true for $n=1$
  \begin{align}
    \<\phi_1|Q_1|\phi_1^\bot\>
    = \frac1d \<\psi_1\psi_1^\bot|V|\psi_1^\bot\psi_1\>
    = \frac1d = - \frac12 (0 - \frac12).
  \end{align}
  Suppose it is true for some $n$, let us show it also holds for
  $n+1$. Without loss of generality we can assume $\phi_1$ and
  $\phi_1^\bot$ are orthogonal on one of the first $n$ copies thus we
  can write
  \begin{align}
    |\phi_1\> = |\phi\>|\psi\psi^*\>, \quad
    |\phi_1^\bot\> = |\phi^\bot\>|\tilde\psi\tilde\psi^*\>.
  \end{align}
  Then by using recursive formula \eqref{eq:Qn-recursive} we have
  \begin{align}
    \nonumber
    &\<\phi_1|Q_{n+1}|\phi_1^\bot\> \\
    &= \<\phi|Q_n|\phi^\bot\> \left(
      \<\psi\psi^*|\tilde\psi\tilde\psi^*\>
      - 2\<\psi\psi^*|Q_1|\tilde\psi\tilde\psi^*\>
    \right) \\
    &= - \frac12 \prod_{i=1}^n \left(
      |\<\psi_i|\tilde\psi_i\>|^2 - \frac12
    \right)\left(
      |\<\psi|\tilde\psi\>|^2
      - \frac2d\<\psi\tilde\psi|V|\tilde\psi\psi\>
    \right) \\
    &= - \frac12 \prod_{i=1}^{n+1} \left(
      |\<\psi_i|\tilde\psi_i\>|^2 - \frac12\right).
  \end{align}
  It is evident that to maximize \eqref{eq:r1-max-superpos},
  i.e. obtain $\frac12$, one needs $p=\frac12$ and
  \eqref{eq:r1-max-coh} equal to $2^{-(n+1)}$.  This requires
  $\left||\<\psi_i|\tilde\psi_i\>|^2-\frac12\right|=\frac12$ for all
  $i$, that is $\psi_i$ and $\tilde\psi_i$ must be equal or orthogonal
  and further for \eqref{eq:r1-max-coh} to be positive they must be
  orthogonal on odd number of copies and equal on the rest.
\end{proof}

\subsection{Digression: half-property for a class of states $\phi_2$ via $Q^\Gamma$}

We consider the following class of states
\begin{align}
  \label{eq:foo-2}
  |\phitwo\>=a|e_1\>|e_1^*\>+b|e_2\>|e_2^*\>
\end{align}
with $a,b\geq 0$, $|e_1\>\perp |e_2\>$.  In the state--operator
isomorphism they correspond to positive matrices $C_{AA'}$ (see sect
\ref{sec:normal}).  Then $C_A$ and $C_{A'}$ are also positive, hence
normal, so that it is a subclass of states for which we have proved
\halfp{} in section \ref {sec:normal}.  Here we present another proof
for this class of states \eqref{eq:foo-2}.  (In section \ref{sec:ent}
we present a third proof, which uses principle of noincreasing
entanglement by LOCC).

We can write
\begin{align}
  \<\phitwo|Q|\phitwo\>=\tr (Q^\Gamma P_\phitwo^\Gamma)
\label{eq:q-gamma-p}
\end{align}
with $P_\phitwo=|\phitwo\>\<\phitwo|$.
We have
\begin{align}
  P_\phitwo^\Gamma
  =a^2 P_{|e_1\>|e_1\>}+ b^2 P_{|e_2\>|e_2\>}
  +ab(P_{\psi_+}-P_{\psi_-})
\end{align}
with
\begin{align}
  |\psi_\pm\>={1\over \sqrt2}(|e_1\>|e_2\>\pm|e_2\>|e_1\>).
\end{align}
Now recall that
\begin{align}
  Q^\Gamma={3\over 8} P_s\ot P_s - {5\over 8} P_a\ot P_a
  +{1\over 8}(P_a\ot P_s +P_s\ot P_a).
\label{eq:q-gamma}
\end{align}
Note that vectors $|e_1\>|e_1\>$, $|e_2\>|e_2\>$ as well as $\psi_+$
lie in the symmetric subspace i.e. $P_s\ot P_s +P_a\ot P_a$, while
$\psi_-$ lies in the antisymmetric subspace $P_s\ot P_a +P_a\ot P_s$.
Therefore, one can estimate the expression (\ref{eq:q-gamma-p}) from
above, by assuming, that triplet states lie solely within $P_s \ot
P_s$, obtaining
\begin{align}
  \<\phitwo|Q|\phitwo\>=\tr (Q^\Gamma P_\phitwo^\Gamma)\leq
  {3\over 8}(a^2 + b^2 +ab) -{1\over 8} ab\leq \frac12.
\end{align}

\section{Bounds for maximal overlap with Q for all
  states $\phitwo$.}
\label{sec:bounds}

In this section we show that we can improve the bound obtained by
means of product states in the previous section.

\subsection{Strictly less than 3/4}

In the previous section we have provided the following bound
\begin{align}
  \label{eq:le-34}
  \sup_{\phi_2} \<\phi_2|Q|\phi_2\> \le \frac34.
\end{align}
Let us now show that the bound cannot be tight. To this end assume
that we have equality.  Let us recall the bound of
\eqref{eq:sup:pQp+pQo} on the overlap of rank two states with $Q$
\begin{align}
  \label{eq:bound-by-r1+coh}
  \sup_{\phi_2} \<\phi_2|Q|\phi_2\>
  \le \sup_{\phi_1,\phi_1^\bot} (\pQp + |\pQo|).
\end{align}
Our assumption thus implies that RHS $\ge\frac34$.  As
$\pQp\le\frac38$ this requires
\begin{align}
  \label{eq:ge-frac38}
  |\re\pQo| \ge \frac38
\end{align}
and by Schwarz inequality both $\phi_1$ and $\phi_1^\bot$ must have
maximal projection on $Q$ which through proposition \ref{pp:r1-form}
implies they must be of the form $|xx^*\>_{AB}|yy^*\>_{A'B'}$. However
for two such orthogonal states by direct calculations we obtain
\begin{align}
  |\re\<\phi_1|Q|\phi_1^\bot\>| \le \frac18
\end{align}
which is in contradiction with \eqref{eq:ge-frac38} and hence with our
assumption of equality in \eqref{eq:le-34}.  Thus we obtain
\begin{align}
  \sup_{\phi_2} \<\phi_2|Q|\phi_2\> < \frac34.
\end{align}

Numerical optimization suggests the bound \eqref{eq:bound-by-r1+coh}
is actually equal to $17 \over 32$.  If we want to optimize
independently both terms of the bound \eqref{eq:bound-by-r1+coh} we
get
\begin{align}
  \label{eq:bound-38+}
 \sup_{\phi_2} \<\phi_2|Q|\phi_2\>
 \le \frac38 + \sup_{\phi_1,\phi_1^\bot} |\pQo|
\end{align}
which numerically gives $5 \over 8$.  At the moment we do not
have analytical proofs of these estimates.

\subsection{Beyond 3/4}

We have seen that product states attaining maximum overlap with $Q$
have to be of the form $|\phi\>=|x\>_A|x^*\>_B|y\>_{A'}|y^*\>_{B'}$,
i.e. the partial transpose of $\phi$ belongs to the product of symmetric
subspaces.  From continuity, if the overlap of $\phi$ with $Q$ is
close to maximal, the state $\phi$ should have big overlap with states
of the above form.  Here we provide quantitative estimate.  First we
will show that in such case $\phi$ has big overlap with $P_s\ot P_s$:

\begin{lemma}
  For states $\phi$ product with respect to $AA':BB'$ cut we have
  \begin{align}
    \<\phi|P_s^{AB}\ot P_s^{A'B'}|\phi\>\geq
    4 \<\phi^\Gamma|Q|\phi^\Gamma\> -{1\over 2},
  \end{align}
  where action $\Gamma$ is well defined because $\phi$ is product.
\end{lemma}

\begin{proof}
  It follows from the formula (\ref{eq:q-gamma}) and a bit of algebra.

  We then have that large overlap of a product state $\phi$ with $P_s\ot
  P_s$ implies large overlap with vectors of the form $|xxyy\>$.
\end{proof}

\begin{lemma}
  For all states $\phi$ product with respect to $AA':BB'$ cut we have
  \begin{align}
    \sup_{x,y}|\<\phi|xx\>_{AB}|yy\>_{A'B'}|^2 \geq
    4\<\phi|P^s_{AB}\ot P^s_{A'B'}|\phi\>-3.
  \end{align}
\end{lemma}

\begin{proof}
  Write $|\phi\>=|e\>_{AA'}|f\>_{BB'}$.  We then find
  \begin{multline}
    \<\phi|P_s^{AB}\ot P_s^{A'B'}|\phi\> \\=
    {1\over 4} (1+ \tr \varrho^e_{A} \varrho^f_{B}+
    \tr \varrho^e_{A'} \varrho^f_{B'}+|\<e|f\>|^2)
  \end{multline}
  where $\varrho^e_{A}$ is reduced density matrix of $|e\>$ etc.
  Schwarz inequality then implies
  \begin{multline}
    \<\phi|P_s^{AB}\ot P_s^{A'B'}|\phi\> \\ \leq
    {1\over 4}(1+2 \max(\tr \varrho_e^2,\tr \varrho_f^2)+
    |\<e|f\>|^2)
    \label{eq:phi-ss}
  \end{multline}
  where $\varrho_{e}$ is either of reduced density matrices of $|e\>$,
  similarly for $\varrho_f$.

  On the other hand one finds
  \begin{align}
    |\<\phi|xxyy\>|&=|\<e|xy\>\<f|xy\>| \\
    &\geq |\<e|xy\>\<f|e\>\<e|xy\>|
    =|\<e|xy\>|^2 |\<e|f\>|
    \label{eq:phixxyy}
  \end{align}
  which implies
  \begin{align}
    \sup_{x,y}|\<\phi|xxyy\>|^2\geq \max(p_e,p_f)|\<e|f\>|
  \end{align}
  where $p_e,p_f$ are the largest eigenvalues of $\varrho_e,\varrho_f$
  respectively.  Combining the two equations, and noticing that without
  loss of generality one can assume that $\tr \varrho_e^2=p_e^2 +(1-p_e)^2$
  and the same for $\tr \varrho_f^2$, one obtains
  \begin{align}
    \sup_{x,y}|\<\phi|xxyy\>|^2\geq {1\over 4}(1+\alpha^2)\beta
  \end{align}
  and
  \begin{align}
    \<\phi|P_s\ot P_s|\phi\>\leq {1\over 4}(2+\alpha^2+\beta)
\end{align}
  where
  \begin{align}
    \alpha=\sqrt{2\max(\tr\varrho_e^2,\tr \varrho_f^2-1)};\quad
    \beta=|\<e|f\>|^2;\quad 0\leq \alpha,\beta\leq 1.
  \end{align}
  Treating $\alpha$ and $\beta$ as independent variables, after some
  elementary, but lengthy algebra, one gets the desired result.
\end{proof}

The above lemmas lead to the following

\begin{proposition}
  For any product state $\phi$ we have
  \begin{align}
    \sup_\chi|\<\phi|\chi\>|^2\geq 16 \<\phi|Q|\phi\>-5
  \end{align}
  where supremum is taken over vectors
  $\chi=|x\>_A|x^*\>_B|y\>_{A'}|y^*\>_{B'}$.
  \label{prop:phi-chi}
\end{proposition}

Subsequently, writing
\begin{align}
  \phi=a \chi + b \psi;\quad \phi^\perp=\tilde a \tilde\chi +
  \tilde b \tilde\psi
\end{align}
where $\phi^\perp$ is a product state orthogonal to $\phi$, and
$\chi\perp\psi$, $\tilde \chi\perp \tilde \psi$, with
$\chi,\tilde\chi$ being of the form $|xx^*yy^*\>$ and $\psi$,
$\tilde\psi$ normalized, we obtain
\begin{align}
  |\<\phi|Q|\phi^\perp\>|\leq |a\tila|\, |\<\chi|Q|\tilde{\chi}\>|+
  \sqrt{3\over 8}(|a\tilb| + |b\tila|)+ |b\tilb|
  \label{eq:q-phi}
\end{align}
where we have used the fact that maximal overlap of $Q$ with a product
state does not exceed $3/8$.  By direct computation we also obtain
\begin{align}
  \<\chi|Q|\tilde\chi\>=-{1\over 8} + {1\over 4} (\<\chi_1|\tilde\chi_1\>
  +\<\chi_2|\tilde\chi_2\>)
  \label{eq:q-chi}
\end{align}
where $|\chi_1\>=|xx^*\>_{AB},|\chi_2\>=|yy^*\>_{A'B'}$ and
$|\tilde\chi_1\>=|\tilde x{\tilde x}^*\>_{AA'},
|\tilde\chi_2\>=|\tilde y{\tilde y}^*\>_{BB'}$.  Using the fact that
$\<\phi|\phi^\perp\>=0$ we get
\begin{align}
  |\<\chi_1|\tilde\chi_1\>|\,|\<\chi_2|\tilde\chi_2\>|
  \leq |b\tila| + |a \tilb|.
  \label{eq:chichi}
\end{align}
Since for any numbers $a$, $b$ satisfying $0 \le a,b\leq 1$ we have
$a+b\leq ab+1$ and combining (\ref{eq:q-phi}), (\ref{eq:q-chi}) and
(\ref{eq:chichi}) we get

\begin{proposition}
  For any product orthogonal states $\phi$ and $\phi^\perp$ we have
  \begin{align}
    |\<\phi|Q|\phi^\perp\>|&\leq
    a_1a_2(-{1\over 8}+{1\over 4}(1+a_1 b_2 + a_2 b_1)) \nonumber\\
    &\phantom{\leq}+\sqrt{3\over 8}
    (a_1 b_2 + a_2 b_1)+b_1 b_2\equiv g(a_1,a_2)
  \end{align}
  where $a_1=|a|=|\<\phi|\chi\>|$, $a_2=|\tila|=|\<\phi|\chi\>|$,
  $b_1=\sqrt{1-a_1^2}$, $b_2=\sqrt{1-a_2^2}$, and $\chi,\tilde\chi$
  are of the form $|xx^*yy^*\>$.
\end{proposition}

Let us observe that
\begin{align}
  &\sup_{\phi_2} \<\phi_2|Q|\phi_2\> \\&= \sup_{\phi_1,\phi_1^\bot,p}
  (\sqrt {p} \<\phi_1| + \sqrt {1-p}\<\phi_1^\bot|) Q
  (\sqrt {p} |\phi_1\> + \sqrt {1-p}|\phi_1^\bot\>) \\
  &= \sup_{\phi_1,\phi_1^\bot}\sup_p \nonumber\\&\quad
  \begin{bmatrix}
    \small \sqrt {p} \\ \sqrt {1-p}
  \end{bmatrix}^T
  \begin{bmatrix} \label{eq:mat}
    \pQp & \re \<\phi_1|Q|\phi_1^\bot\>\\
    \re \<\phi_1^\bot|Q|\phi_1\> & \oQo
  \end{bmatrix}
  \begin{bmatrix}
    \sqrt {p} \\ \sqrt {1-p}
  \end{bmatrix}
  \\
  \label{eq:fancy-sup}
  &=\sup_{\phi_1,\phi_1^\bot}\frac12 \bigg( \pQp + \oQo
  \\&\phantom{=}
  + \sqrt {(\pQp-\oQo)^2 + 4 (\re \<\phi_1|Q|\phi_1^\bot\>)^2}
  \bigg)
\end{align}
the last expression is simply larger eigenvalue of the matrix in
\eqref{eq:mat}.

Now denoting $\gamma_1=\<\phi|Q|\phi\>$,
$\gamma_2=\<\phi^\perp|Q|\phi^\perp\>$, we get
\begin{align}
  \<\phitwo|Q|\phitwo\>\leq \gamma_1+\gamma_2
\end{align}
from Schwarz inequality. On the other hand using \eqref{eq:bound-38+}
and proposition \ref{prop:phi-chi} we get
\begin{align}
  \<\phitwo|Q|\phitwo\>\leq \frac38 + \sup_{a_1,a_2} g(a_1,a_2)
\end{align}
where supremum is taken over $a_1$, $a_2$ satisfying
\begin{align}
  16 \gamma_i -5 \le a_i^2\le 1 ,\quad i=1,2.
\end{align}

Finally we obtain the following estimate
\begin{align}
  \tQt \le \frac38 + \min(\gamma, f(\gamma))
\end{align}
where $\gamma=\min(\gamma_1, \gamma_2)$ and
\begin{align}
  f(\gamma) = \sup_{a_1,a_2} g(a_1,a_2)
\end{align}
where supremum is taken over $16\gamma-5 \le a_i^2 \le 1$.  Looking on
the plot of $g(a_1,a_2)$ one can find that the maximum is obtained for
$a_1=a_2$.  This leads to the bound
\begin{align}
  \<\phitwo|Q|\phitwo\>\leq 0.74971<3/4.
\end{align}

\section{Application of Entanglement measures}
\label{sec:ent}

Then we will show how entanglement measures can be applied to the
problem of \halfp{}.

The formula $\<\phitwo|Q|\phitwo\>$ can be written as follows:
\begin{align}
  \<\phitwo|Q|\phitwo\> = \tr (\tcal(|\phitwo\>\<\phitwo|) Q)
\end{align}
where $\tcal$ is pairwise $UU^*$ twirling, followed by random
permutation of pairs.  Since $\tcal$ is LOCC operation, the state
$\sigma=\tcal(|\phitwo\>\<\phitwo|)$ cannot have greater entanglement
than the state $\phitwo$. Then, one can hope, that if entanglement of
$\sigma$ is not too large, then also $\tr \sigma Q$ will be bounded.
Write
\begin{align}
  \label{eq:multi-iso}
  \sigma &= {p\over 2} (\portn\ot\pplus +\pplus\ot \portn) +
  s \pplus\ot \pplus \nonumber\\&\phantom{=}
  + (1-p-s) \portn\ot \portn
\end{align}
with $\portn=(I-\pplus)/(d^2-1)$ and probabilities $p,s$ satisfying
$p+s\leq 1$. Then we have
\begin{align}
  \tr \sigma Q = p.
\end{align}

\subsection{Negativity}

We will use the negativity \cite{ZyczkowskiHSP-vol}, or more precisely a
closely related quantity $\|\varrho^\Gamma\|$, which is monotonous under
LOCC \cite{Vidal-Werner}.  In our case, one finds that
\begin{align}
  \|\sigma^\Gamma\|= {1\over 4}(2|1-16s|+|1+8s-4p|+1+24 s +4p).
\label{eq:neg}
\end{align}
Now monotonicity requires that
\begin{align}
  \|\sigma^\Gamma\|\leq \|\phitwo^\Gamma\|=|a+b|^2
\end{align}
where $a,b$ are Schmidt coefficients of $\phitwo$.  This inequality
together with (\ref{eq:neg}) implies in particular that
\begin{align}
  p\leq {1\over 4} - 6 \<\phitwo|\pplus\ot \pplus|\phitwo\>
  +2|a+b|^2.
\end{align}
Note that for fixed Schmidt coefficients $a,b$ maximal overlap with
$\pplus\ot \pplus$ cannot exceed $|a+b|^2/16$.  We then obtain, that
for those states which achieve this maximal overlap there holds the
half-property.  However such states are simply states of the form
\begin{align}
  \phitwo=a|e_1\>_{AA'}|e_1^*\>_{BB'}+b|e_2\>_{AA'}|e_2^*\>_{BB'}
\end{align}
with $a,b,\geq 0$.  Since such states have positive matrix $C$ we end
up with yet another proof of \halfp{} for this class of states.

For states that are orthogonal to $P_+\ot P_+$ negativity gives bound
$3/4$. We have also tried the relative entropy of entanglement and
the realignment but worse results have been obtained.

\subsection{Half-property and Schmidt rank of some symmetric states}

The possibility of application of entanglement measures to the problem
of \halfp{} can be also seen from the following different perspective.
Namely, one can classify states with respect to Schmidt rank.  We say
that a mixed state has Schmidt rank $k$, if it can be written as a
mixture of pure states of Schmidt rank $k$, but cannot be written as a
mixture of pure states of Schmidt rank $k-1$ (cf
\cite{Terhal-Pawel-rank}).  We then have the following

\begin{fact}
  The projector $Q$ has \halfp{} if and only if for all states
  $\sigma$ of the form (\ref{eq:multi-iso}) which have Schmidt rank
  $\leq 2$ we have $p \leq 1/2$.
\end{fact}

One direction is trivial, the other follows from twirling.  Thus if we
are able to prove that all states $\sigma$ of the form
(\ref{eq:multi-iso}) with $p>1/2$ have Schmidt rank $>2$, we would
solve the problem of \halfp.  To this end we should find a map
$\Lambda$ such that $I\ot \Lambda$ is nonnegative on Schmidt rank two
pure states (such maps are called two positive), and at the same time
negative on all states $\sigma$ with $p\geq 1/2$.  Indeed, this would
mean that all states $\sigma$ with $p\geq 1/2$ have Schmidt rank $>2$.

Using this approach one can also get bounds for our quantity
$\<\phitwo| Q |\phitwo\>$.  For example we have checked that the
following two-positive map $\Lambda(A)= I \,\tr A - 1/2 A$ is negative
for $p> 3/4$ which reproduces the bound obtained by means of product
states.

In this context we see why entanglement measures can be applied to our
problem.  Namely, if an entanglement measure of a given state is greater
than maximum of this measure over Schmidt rank two pure states, then
the state must have Schmidt rank two greater than 2.

\subsection{Continuity of entanglement and bound entanglement}

One could ask the question whether there exist a continuous
entanglement measure which would detect between three kinds of states:
1) separable, 2) bound entangled, and 3) distillable ones.  There are
measures such as the entanglement of formation which distinguish between 1
(for which it is zero) versus 2 and 3 (for which it is nonzero), and
there is a measures, the distillable entanglement, which distinguishes
between 1 and 2 (for which it is zero) versus 3 (for which it is non
zero).  But any measure that would distinguish between the three
classes of states by its value in a way that entanglement of all bound
entangled states is non zero but smaller than entanglement of any
distillable state must be non continuous.  Indeed for such a measure
there must be a range of values reserved for bound entangled states,
creating a gap between separable states and distillable ones. On the
other hand we can take a sequence of distillable states with a limit
being a separable state (and so with zero value of entanglement), but the
limit of the entanglement for this sequence must be at most supremum
of its value on bound entangled states.  Note that provided that NPT
bound entangled states exist such a measure would also increase under
tensoring because then there would exist bound entangled states whose
tensor product is distillable \cite{ShorST2001}, as a matter of fact
the same would then hold for the distillable entanglement.

\section*{ACKNOWLEDGEMENTS}

This work is supported by EU grant SCALA FP6-2004-IST no.015714.

\section*{APPENDIX}

\begin{lemma}
  \label{lm:equal-divide}
  The minimum value of $\sum_{i=1}^d |\tila_i|^2$ subject to
  $\sum_{i=1}^{d} \tila_i = z$ where $\tila_i, z \in \mathbb{C}$ is
  obtained by settings $\tila_i = \frac{z}{d}$.
\end{lemma}

\begin{proof}
  From the parallelogram identity we have
  \begin{align}
    \frac12 |\tila_i + \tila_j|^2
    = |\tila_i|^2 + |\tila_j|^2 - \frac12 |\tila_i - \tila_j|
    \le |\tila_i|^2 + |\tila_j|^2
  \end{align}
  with equality iff $\tila_i = \tila_j$.  Thus whenever for some
  $\tila_i$, $\tila_j$ we have $\tila_i\neq \tila_j$ we can replace
  them with two instances of $\frac{\tila_i + \tila_j}{2}$ decreasing
  the value of $\sum_{i=1}^d |\tila_i|^2$ and leaving the constrain
  satisfied.  This implies that the optimal solution is to take all
  $\tila_i$ equal, i.e. $\tila_i = \frac{z}{d}$.
\end{proof}

\begin{proposition}
  \label{pp:a1b1b2}
  For all $d\ge3$ dimensional vectors $\vec{a}$ and $\vec{b}$ with
  complex elements $\tila_i$ and $\tilb_i$ and satisfying the
  constraints
  \begin{align}
    \label{eq:a1b1b2-constraint}
    \sum_{i=1}^{d} \tila_i = \sum_{i=1}^{d} \tilb_i = 0, \qquad
    \sum_{i=1}^{d} |\tila_i|^2 + \sum_{i=1}^{d}|\tilb_i|^2 = \frac{1}{d}
  \end{align}
  the following equality holds
  \begin{align}
    \label{eq:a1b1b2-f}
    \max_{\vec{a},\vec{b}} \left(
      |\tila_1 + \tilb_1|^2 + |\tila_1+\tilb_2|^2
    \right) = \frac{3d-4}{d^2}.
  \end{align}
\end{proposition}

\begin{corollary}
  For $d=4$ under this constraints we have
  \begin{align}
    \max_{\vec{a},\vec{b}} \left(
      |\tila_1 + \tilb_1|^2 + |\tila_1+\tilb_2|^2
    \right)= \frac12.
  \end{align}
\end{corollary}

\begin{proof}[Proof of proposition \ref{pp:a1b1b2}]
  We denote function \eqref{eq:a1b1b2-f} as $f$, the vector of all
  $\tila_i$ as $\vec{a}$, the vector of all $\tilb_i$ as $\vec{b}$,
  and we use their polar decompositions
  \begin{align}
    \tila_i = a_i e^{i\alpha_i}, \quad
    \tilb_i = b_i e^{i\beta_i}, \quad
    a_i,b_i\in\mathbb{R}.
  \end{align}

  In optimizing function $f$ under the constraints
  \eqref{eq:a1b1b2-constraint} we shrink the set of possible $\vec{a}$
  and $\vec{b}$ in such a way to simplify the form of $f$ and the
  constraints but keeping at least one of the global maxima within the
  shrinking set.

  \begin{enumerate}
  \item Without loss of generality we can take $\tila_1=a_1\ge0$.
    Thus we optimize
    \begin{align}
      f(\vec{a},\vec{b})
      &= |a_1 + \tilb_1|^2 + |a_1+\tilb_2|^2 \\
      &= 2a_1^2 + b_1^2 + b_2^2 + 2a_1(b_1\cos\beta_1 + b_2\cos\beta_2).
    \end{align}

  \item We can consider only $\vec{b}$ for which
    \begin{align}
      b_1\cos\beta_1 + b_2\cos\beta_2 \ge 0.
    \end{align}
    (If it is negative we can change its sign by multiplying $\vec{b}$
    by $e^{i\pi}$ and thus increase $f$).

  \item In maximizing $f$ under the constraints it is always best to set
    \begin{align}
      \tila_i &= - \frac{a_1}{d-1} &(i>1) \\
      \tilb_i &= - \frac1{d-2}(\tilb_1 + \tilb_2) &(i>2)
    \end{align}
    Indeed whenever this setting is not used we can by lemma
    \ref{lm:equal-divide} obtain some freedom in the second constraint
    which we can use to increase $a_1$ and one of $b_1$ or $b_2$
    without decreasing $f$.  Thus it is enough to consider $\vec{a}$
    and $\vec{b}$ satisfying this setting, i.e. we optimize function
    $f(a_1,\tilb_1,\tilb_2)$ subject to the following constraints
    \begin{multline}
      \frac{d}{d-1} a_1^2 + b_1^2 + b_2^2
      + \frac{1}{d-2} \left|\tilb_1 + \tilb_2\right|^2 = \frac1d, \\
      a_1 \ge 0, \quad b_1\cos\beta_1 + b_2\cos\beta_2 \ge 0.
    \end{multline}

  \item Further we show that it is enough to consider
    $\tilb_1,\tilb_2\in\mathbb{R}$ as replacing $\tilb_1$ with
    $\tilb_1'=b_1\cos\beta_1$ and $\tilb_2$ with
    $\tilb_2'=b_2\cos\beta_2$ and changing $a_1$ to $a_1'$ to fit the
    constraint does not decrease $f$, i.e. $f(a_1', \tilb'_1,
    \tilb'_2) \ge f(a_1, \tilb_1, \tilb_2)$.  Namely we have
    \begin{multline}
      f(a_1', \tilb'_1, \tilb'_2) =
      2a_1'^2 + b_1^2\cos^2\beta_1 + b_2^2\cos^2\beta_2 \\
      + 2a_1'(b_1\cos\beta_1 + b_2\cos\beta_2)
    \end{multline}
    and the main constraint is
    \begin{multline}
      \frac{d}{d-1} a_1'^2 + b_1^2\cos^2\beta_1 + b_2^2\cos^2\beta_2 \\
      + \frac{1}{d-2} \left|b_1\cos\beta_1 + b_2\cos\beta_2\right|^2
      = \frac1d.
    \end{multline}
    First we show that $a_1' \ge a_1$ which is evident from the
    difference of main constraints
    \begin{align}
      &\frac{d}{d-1} (a_1'^2 - a_1^2) \nonumber\\&=
      b_1^2\sin^2\beta_1 + b_2^2\sin^2\beta_2
       \nonumber\\ &\qquad
      + \frac{1}{d-2} \left(
        \left|b_1e^{i\beta_1} + b_2e^{i\beta_2}\right|^2 -
        \left|b_1\cos\beta_1 + b_2\cos\beta_2\right|^2
      \right)
      \nonumber\\ &\ge 0.
    \end{align}
    Next we use this difference to show that $f$ does not decrease
    after the replacement
    \begin{multline}
      f(a_1', \tilb'_1, \tilb'_2) - f(a_1, \tilb_1, \tilb_2) \\
      = 2(a_1'^2 - a_1^2) - b_1^2 \sin^2\beta_1 - b_2^2\sin^2\beta_2 \\
      + 2(a_1' - a_1)(b_1\cos\beta_1 + b_2\cos\beta_2) \\
      \ge \frac{d-2}{d}(b_1^2 \sin^2\beta_1 + b_2^2\sin^2\beta_2) \ge 0.
    \end{multline}

    So we can focus on a problem with $\tilb_1,\tilb_2\in\mathbb{R}$
    \begin{align}
      &f(a_1, b_1, b_2) = 2a_1^2 + b_1^2 + b_2^2 + 2a_1(b_1+b_2) \\
      &\frac{d}{d-1} a_1^2 + b_1^2 + b_2^2 + \frac{1}{d-2} (b_1 + b_2)^2
      = \frac1d, \nonumber\\&\qquad\qquad
      a_1 \ge 0, \quad b_1 + b_2 \ge 0.
    \end{align}

  \item In analogous way we show that it is enough to consider
    $b_1=b_2 \ge 0$ as taking $b_1'=b_2'=\frac{|b_1+b_2|}{2}$ and
    changing $a_1$ to $a_1'$ to fit the constraint does not decrease
    $f$. Then the optimization simplifies to
    \begin{align}
      &f(a_1, b_1) = 2(a_1 + b_1)^2 \\
      &\frac{d}{d-1}a_1^2 + \frac{2d}{d-2}b_1^2 = \frac1d,
      \quad a_1, b_1 \ge 0.
    \end{align}

  \item We compute $b_1$ from the constraint and substitute to $f$
    which gives
    \begin{align}
      &f(a_1) = 2\left(a_1 + \sqrt{ x - ya_1^2}\right)^2 \\
      &a_1 \in \left[0, \sqrt{x / y}\right]
    \end{align}
    where
    \begin{align}
      x = \frac{d-2}{2d^2}, \qquad y = \frac{d-2}{2(d-1)}.
    \end{align}

    Function $f$ has its maximum when the expression in the
    parenthesis has the maximum (as it is nonnegative). We consider
    its derivative
    \begin{align}
      \frac{\partial}{\partial a_1} \left(a_1 + \sqrt{ x - ya_1^2}\right) &=
      1 - \frac{ya_1}{\sqrt{x - ya_1^2}}
    \end{align}
    which is zero for
    \begin{align}
      a_1^\star = \sqrt{\frac{x}{y^2 + y}}
    \end{align}
    and the second derivative is negative in $a_1^\star$ so the
    maximum is equal to
    \begin{align}
      f(a_1^\star)
      &= 2\left(\sqrt{\frac{x}{y^2 + y}} + \sqrt{\frac{xy}{y+1}}\right)^2 \\
      &= 2x(y^{-1} + 1) = \frac{3d - 4}{d^2}
    \end{align}
    The global maximum could also be on one of the boundaries but for
    $d\ge3$ $f(a_1^\star)$ is always greater than the values on the
    boundaries.
  \end{enumerate}
\end{proof}


\begin{thebibliography}{10}
\providecommand{\url}[1]{#1}
\csname url@samestyle\endcsname
\providecommand{\newblock}{\relax}
\providecommand{\bibinfo}[2]{#2}
\providecommand{\BIBentrySTDinterwordspacing}{\spaceskip=0pt\relax}
\providecommand{\BIBentryALTinterwordstretchfactor}{4}
\providecommand{\BIBentryALTinterwordspacing}{\spaceskip=\fontdimen2\font plus
\BIBentryALTinterwordstretchfactor\fontdimen3\font minus
  \fontdimen4\font\relax}
\providecommand{\BIBforeignlanguage}[2]{{%
\expandafter\ifx\csname l@#1\endcsname\relax
\typeout{** WARNING: IEEEtran.bst: No hyphenation pattern has been}%
\typeout{** loaded for the language `#1'. Using the pattern for}%
\typeout{** the default language instead.}%
\else
\language=\csname l@#1\endcsname
\fi
#2}}
\providecommand{\BIBdecl}{\relax}
\BIBdecl
\providecommand{\eprint}[1]{\href{http://arxiv.org/abs/#1}{#1}}
\providecommand{\eprintNEW}[2]{\href{http://arxiv.org/abs/#1}{arXiv:#1}}

\bibitem{rmp-ent}
R.~Horodecki, P.~Horodecki, M.~Horodecki, and K.~Horodecki, ``Quantum
  entanglement,'' \emph{Rev. Mod. Phys.}, vol.~81, pp. 865--942, 2009,
  \eprint{quant-ph/0702225}.

\bibitem{BBPSSW96}
C.~H. Bennett, G.~Brassard, S.~Popescu, B.~Schumacher, J.~A. Smolin, and W.~K.
  Wootters, ``Purification of noisy entanglement and faithful teleportation via
  noisy channels,'' \emph{Phys. Rev. Lett.}, vol.~76, pp. 722--725, 1996.

\bibitem{BDSW1996}
C.~H. Bennett, D.~P. DiVincenzo, J.~A. Smolin, and W.~K. Wootters,
  ``Mixed-state entanglement and quantum error correction,'' \emph{Phys. Rev.
  A}, vol.~54, pp. 3824--3851, 1996, \eprint{quant-ph/9604024}.

\bibitem{HHH1997-distill}
M.~Horodecki, P.~Horodecki, and R.~Horodecki, ``Inseparable two
  s\mbox{pin-$\frac12$} density matrices can be distilled to a singlet form,''
  \emph{Phys. Rev. Lett.}, vol.~78, pp. 574--577, 1997.

\bibitem{Vidal-cost2002}
G.~Vidal, W.~D{\"u}r, and J.~I. Cirac, ``Entanglement cost of bipartite mixed
  states,'' \emph{Phys. Rev. Lett.}, vol.~89, p. 027901, 2002,
  \eprint{quant-ph/0112131}.

\bibitem{YangHHS2005-cost}
D.~Yang, M.~Horodecki, R.~Horodecki, and B.~Synak-Radtke, ``Irreversibility for
  all bound entangled states,'' \emph{Phys. Rev. Lett.}, vol.~95, p. 190501,
  2005, \eprint{quant-ph/0506138}.

\bibitem{Terhal1}
B.~M. Terhal, ``Is entanglement monogamous?'' \emph{IBM J. Res. Develop.},
  vol.~48, p.~71, 2004.

\bibitem{termo}
P.~Horodecki, R.~Horodecki, and M.~Horodecki, ``Entanglement and
  thermodynamical analogies,'' \emph{Acta Phys. Slovaca}, vol.~48, p. 141,
  1998, \eprint{quant-ph/9805072}.

\bibitem{thermo-ent2002}
M.~Horodecki, J.~Oppenheim, and R.~Horodecki, ``Are the laws of entanglement
  theory thermodynamical?'' \emph{Phys. Rev. Lett.}, vol.~89, p. 240403, 2002,
  \eprint{quant-ph/0207177}.

\bibitem{activation}
P.~Horodecki, M.~Horodecki, and R.~Horodecki, ``Bound entanglement can be
  activated,'' \emph{Phys. Rev. Lett.}, vol.~82, pp. 1056--1059, 1999,
  \eprint{quant-ph/9806058}.

\bibitem{VollbWolf_activation}
K.~G.~H. Vollbrecht and M.~M. Wolf, ``Activating distillation with an
  infinitesimal amount of bound entanglement,'' \emph{Phys. Rev. Lett.},
  vol.~88, p. 247901, 2002, \eprint{quant-ph/0201103}.

\bibitem{Masanes1_activation}
L.~Masanes, ``All bipartite entangled states are useful for information
  processing,'' \emph{Phys. Rev. Lett.}, vol.~96, p. 150501, 2006,
  \eprint{quant-ph/0508071}.

\bibitem{pptkey}
K.~Horodecki, M.~Horodecki, P.~Horodecki, and J.~Oppenheim, ``Secure key from
  bound entanglement,'' \emph{Phys. Rev. Lett.}, vol.~94, p. 160502, 2005,
  \eprint{quant-ph/0309110}.

\bibitem{keyhuge}
------, ``General paradigm for distilling classical key from quantum states,''
  \emph{IEEE Trans. Inf. Theory}, vol.~55, pp. 1898--1929, 2009,
  \eprint{quant-ph/0506189}.

\bibitem{smallkey}
K.~Horodecki, {\L}.~Pankowski, M.~Horodecki, and P.~Horodecki,
  ``Low-dimensional bound entanglement with one-way distillable cryptographic
  key,'' \emph{IEEE Trans. Inf. Theory}, vol.~54, pp. 2621--2625, 2008,
  \eprint{quant-ph/0506203}.

\bibitem{GeneralUncSec}
K.~Horodecki, M.~Horodecki, P.~Horodecki, D.~W. Leung, and J.~Oppenheim,
  ``Quantum key distribution based on private states: unconditional security
  over untrusted channels with zero quantum capacity,'' \emph{IEEE Trans. Inf.
  Theory}, vol.~54, pp. 2604--2620, 2008, \eprint{quant-ph/0608195}.

\bibitem{uncond-prl}
------, ``Unconditional privacy over channels which cannot convey quantum
  information,'' \emph{Phys. Rev. Lett.}, vol. 100, p. 110502, 2008,
  \eprint{quant-ph/0702077}.

\bibitem{bound}
M.~Horodecki, P.~Horodecki, and R.~Horodecki, ``Mixed-state entanglement and
  distillation: Is there a ``bound'' entanglement in nature?'' \emph{Phys. Rev.
  Lett.}, vol.~80, pp. 5239--5242, 1998, \eprint{quant-ph/9801069}.

\bibitem{Peres1}
A.~Peres, ``Separability criterion for density matrices,'' \emph{Phys. Rev.
  Lett.}, vol.~77, pp. 1413--1415, 1996, \eprint{quant-ph/9604005}.

\bibitem{ShorST2001}
P.~W. Shor, J.~A. Smolin, and B.~M. Terhal, ``Nonadditivity of bipartite
  distillable entanglement follows from a conjecture on bound entangled
  {W}erner states,'' \emph{Phys. Rev. Lett.}, vol.~86, pp. 2681--2684, 2001,
  \eprint{quant-ph/0010054}.

\bibitem{WernerPPT}
T.~Eggeling, K.~G.~H. Vollbrecht, R.~F. Werner, and M.~M. Wolf,
  ``Distillability via protocols respecting the positivity of partial
  transpose,'' \emph{Phys. Rev. Lett.}, vol.~87, p. 257902, 2001,
  \eprint{quant-ph/0104095}.

\bibitem{ShorST-superactiv}
P.~W. Shor, J.~A. Smolin, and A.~V. Thapliyal, ``Superactivation of bound
  entanglement,'' \emph{Phys. Rev. Lett.}, vol.~90, p. 107901, 2003,
  \eprint{quant-ph/0005117}.

\bibitem{DuHoCi04}
W.~D{\"u}r, J.~I. Cirac, and P.~Horodecki, ``Nonadditivity of quantum capacity
  for multiparty communication channels,'' \emph{Phys. Rev. Lett.}, vol.~93, p.
  020503, 2004, \eprint{quant-ph/0403068}.

\bibitem{SmithYard}
G.~Smith and J.~Yard, ``Quantum communication with zero-capacity channels,''
  \emph{Science}, vol. 321, p. 1812, 2008, \eprintNEW{0807.4935}.

\bibitem{CzekajPawel}
{\L}.~Czekaj and P.~Horodecki, ``Nonadditivity effects in classical capacities
  of quantum multiple-access channels,'' 2008, \eprintNEW{0807.3977}.

\bibitem{reduction}
M.~Horodecki and P.~Horodecki, ``Reduction criterion of separability and limits
  for a class of distillation protocols,'' \emph{Phys. Rev. A}, vol.~59, pp.
  4206--4216, 1999, \eprint{quant-ph/9708015}.

\bibitem{Werner1989}
R.~F. Werner, ``Quantum states with {E}instein-{P}odolsky-{R}osen correlations
  admitting a hidden-variable model,'' \emph{Phys. Rev. A}, vol.~40, pp.
  4277--4281, 1989.

\bibitem{DiVincenzoSSTT1999-nptbound}
D.~P. DiVincenzo, P.~W. Shor, J.~A. Smolin, B.~M. Terhal, and A.~V. Thapliyal,
  ``Evidence for bound entangled states with negative partial transpose,''
  \emph{Phys. Rev. A}, vol.~61, p. 062312, 2000, \eprint{quant-ph/9910026}.

\bibitem{DurCLB1999-npt-bound}
W.~D{\"u}r, J.~I. Cirac, M.~Lewenstein, and D.~Bru{\ss}, ``Distillability and
  partial transposition in bipartite systems,'' \emph{Phys. Rev. A}, vol.~61,
  p. 062313, 2000, \eprint{quant-ph/9910022}.

\bibitem{Bandyopadhyay2003-n-distil}
S.~Bandyopadhyay and V.~Roychowdhury, ``Classes of $n$-copy undistillable
  quantum states with negative partial transposition,'' \emph{Phys. Rev. A},
  vol.~68, p. 022319, 2003, \eprint{quant-ph/0302093}.

\bibitem{Watrous2003-n-dist}
J.~Watrous, ``Many copies may be required for entanglement distillation,''
  \emph{Phys. Rev. Lett.}, vol.~93, p. 010502, 2004.

\bibitem{Rains1999}
E.~M. Rains, ``Bound on distillable entanglement,'' \emph{Phys. Rev. A},
  vol.~60, pp. 179--184, 1999, \eprint{quant-ph/9809082}.

\bibitem{Clarisse2005-distil}
L.~Clarisse, ``The distillability problem revisited,'' \emph{Quantum Inf.
  Comp.}, vol.~6, pp. 539--560, 2006, \eprint{quant-ph/0510035}.

\bibitem{Clarisse2004-distil-maps}
------, ``Characterization of distillability of entanglement in terms of
  positive maps,'' \emph{Phys. Rev. A}, vol.~71, p. 032332, 2005,
  \eprint{quant-ph/0403073}.

\bibitem{sep1996}
M.~Horodecki, P.~Horodecki, and R.~Horodecki, ``Separability of mixed states:
  Necessary and sufficient conditions,'' \emph{Phys. Lett. A}, vol. 223, p.~1,
  1996, \eprint{quant-ph/9605038}.

\bibitem{TerhalReview}
B.~M. Terhal, ``Detecting quantum entanglement,'' \emph{Journal of Theoretical
  Computer Science}, vol. 287, p. 313, 2002, \eprint{quant-ph/0101032}.

\bibitem{KrausLC}
B.~Kraus, M.~Lewenstein, and J.~I. Cirac, ``Characterization of distillable and
  activatable states using entanglement witnesses,'' \emph{Phys. Rev. A},
  vol.~65, p. 042327, 2002, \eprint{quant-ph/0110174}.

\bibitem{vianna-bound}
R.~O. Vianna and A.~C. Doherty, ``Distillability of {W}erner states using
  entanglement witnesses and robust semidefinite programs,'' \emph{Phys. Rev.
  A}, vol.~74, no.~5, p. 052306, 2006, \eprint{quant-ph/0608095}.

\bibitem{ClarissePhd}
L.~Clarisse, ``Entanglement distillation; a discourse on bound entanglement in
  quantum information theory,'' Ph.D. dissertation, University of York, 2006.

\bibitem{ChattoSarkar}
I.~Chattopadhyay and D.~Sarkar, ``{NPT} bound entanglement- the problem
  revisited,'' 2006, \eprint{quant-ph/0609050}.

\bibitem{SimonNPTbound}
R.~Simon, ``{NPPT} bound entanglement exists,'' 2006,
  \eprint{quant-ph/0608250}.

\bibitem{BrandaoEisert2007}
F.~G.~S.~L. Brand{\~a}o and J.~Eisert, ``Correlated entanglement distillation
  and the structure of the set of undistillable states,'' \emph{J. Math.
  Phys.}, vol.~49, p. 042102, 2008, \eprintNEW{0709.3835}.

\bibitem{lewenstein-2000-47-primer}
M.~Lewenstein, D.~Bru{\ss}, J.~I. Cirac, B.~Kraus, M.~Ku\'s, J.~Samsonowicz,
  A.~Sanpera, and R.~Tarrach, ``Separability and distillability in composite
  quantum systems -a primer-,'' \emph{Journal of Modern Optics}, vol.~47, p.
  2841, 2000, \eprint{quant-ph/0006064}.

\bibitem{AcinVC-mregs}
A.~Ac\'in, G.~Vidal, and J.~I. Cirac, ``On the structure of a reversible
  entanglement generating set for three--partite states,'' \emph{Quantum Inf.
  Comp.}, vol.~3, p.~55, 2003, \eprint{quant-ph/0202056}.

\bibitem{ZyczkowskiHSP-vol}
K.~\.Zyczkowski, P.~Horodecki, A.~Sanpera, and M.~Lewenstein, ``Volume of the
  set of separable states,'' \emph{Phys. Rev. A}, vol.~58, pp. 883--892, 1998,
  \eprint{quant-ph/9804024}.

\bibitem{Vidal-Werner}
G.~Vidal and R.~F. Werner, ``Computable measure of entanglement,'' \emph{Phys.
  Rev. A}, vol.~65, p. 032314, 2002, \eprint{quant-ph/0102117}.

\bibitem{Terhal-Pawel-rank}
B.~M. Terhal and P.~Horodecki, ``Schmidt number for density matrices,''
  \emph{Phys. Rev. A (Rap. Commun.)}, vol.~61, p. 040301, 2000,
  \eprint{quant-ph/9911117}.

\end{thebibliography}


\end{document}